\PassOptionsToPackage{table,dvipsnames}{xcolor}
\documentclass[sigconf,nonacm]{aamas} 

\usepackage[T1]{fontenc}
\usepackage{amsthm}
\usepackage{amsmath}
\usepackage{vote_notation}
\usepackage{mathtools}
\usepackage{tikz}
\usetikzlibrary{patterns}

\usepackage{hyperref}
\usepackage[capitalize]{cleveref}
\usepackage{tcolorbox}
\usepackage{booktabs}

\usepackage{caption}
\usepackage{subcaption}
\usepackage{tikz}
\usetikzlibrary{matrix,decorations.pathreplacing,tikzmark}

\newcommand\fauxsc[1]{\fauxschelper#1 \relax\relax}
\def\fauxschelper#1 #2\relax{%
  \fauxschelphelp#1\relax\relax%
  \if\relax#2\relax\else\ \fauxschelper#2\relax\fi%
}
\def\Hscale{.85}\def\Vscale{.72}\def\Cscale{1.10}
\def\fauxschelphelp#1#2\relax{%
  \ifnum`#1>``\ifnum`#1<`\{\scalebox{\Hscale}[\Vscale]{\uppercase{#1}}\else%
    \scalebox{\Cscale}[1]{#1}\fi\else\scalebox{\Cscale}[1]{#1}\fi%
  \ifx\relax#2\relax\else\fauxschelphelp#2\relax\fi}

\DeclareMathOperator{\sgnText}{\textit{sgn}}
\newcommand{\sgn}[1]{\sgnText({#1})}
\newcommand{\innerprod}[2]{\langle{#1}, {#2}\rangle}

\usepackage{graphicx}
\usepackage{color}

\RequirePackage{thmtools}

\declaretheorem{falseth}
\declaretheorem[sibling=falseth]{theorem}
\declaretheorem[sibling=falseth]{lemma}
\declaretheorem[sibling=falseth]{corollary}
\declaretheorem[sibling=falseth]{claim}

\title[Condorcet Winners and Anscombe’s Paradox Under Weighted Binary Voting]{Condorcet Winners and Anscombe’s Paradox Under Weighted Binary Voting}

\author{Carmel Baharav}
\affiliation{
  \institution{ETH Zürich}
  \city{Zürich}
  \country{Switzerland}}
\email{cbaharav@ethz.ch}

\author{Andrei Constantinescu}
\affiliation{
  \institution{ETH Zürich}
  \city{Zürich}
  \country{Switzerland}}
\email{aconstantine@ethz.ch}

\author{Roger Wattenhofer}
\affiliation{
  \institution{ETH Zürich}
  \city{Zürich}
  \country{Switzerland}}
\email{wattenhofer@ethz.ch}

\begin{abstract}
We consider voting on multiple independent binary issues. In addition, a weighting vector for each voter defines how important they consider each issue.
The most natural way to aggregate the votes into a single unified proposal is \emph{issue-wise majority} (IWM): taking a majority opinion for each issue. 
However, in a scenario known as \emph{Ostrogorski's Paradox}, an IWM proposal may not be a Condorcet winner, or it may even fail to garner majority support in a special case known as \emph{Anscombe's Paradox}.

We show that it is co-NP-hard to determine whether there exists a Condorcet-winning proposal even without weights.
In contrast, we prove that the \emph{single-switch} condition provides an Ostrogorski-free voting domain under identical weighting vectors. We show that verifying the condition can be achieved in linear time and no-instances admit short, efficiently computable proofs in the form of forbidden substructures. On the way, we give the simplest linear-time test for the \emph{voter/candidate-extremal-interval} condition in approval voting and the simplest and most efficient algorithm for recognizing single-crossing preferences in ordinal voting.

We then tackle Anscombe's Paradox. Under identical weight vectors, we can guarantee a majority-supported proposal agreeing with IWM on strictly more than half of the overall weight, while with two distinct weight vectors, such proposals can get arbitrarily far from IWM. The severity of such examples is controlled by the maximum average topic weight $\tilde{w}_{max}$: a simple bound derived from a partition-based approach is tight on a large portion of the range $\tilde{w}_{max} \in (0,1)$. Finally, we extend Wagner's rule to the weighted setting: an average majority across topics of at least $\frac{3}{4}$'s precludes Anscombe's paradox from occurring.
\end{abstract}


\begin{document}

\pagestyle{fancy}
\fancyhead{}

\maketitle 


\section{Introduction}\label{section:motivation}

There are numerous scenarios in which people must decide on a slate of binary issues and come up with a single outcome for each topic. When political parties form a platform, they must aggregate their base's opinions and provide a unified set of stances on numerous separate issues. Similarly, when voters head to the ballot box for local elections, they typically vote yes or no on a series of initiatives. The election results in one outcome for each individual topic. On a smaller scale, a group of flatmates might decide on a series of unrelated topics and generate a plan for living together. For example: Should the kitchen be cleaned once a week or twice a week? Should we get the red couch or the yellow couch?

The most natural way to decide the final outcome in all of these scenarios is to take the majority opinion on each individual topic and aggregate them into a unified party platform, legislative agenda, or roommate contract. However, this approach can yield a surprisingly undesirable outcome: a majority of the voters may actually be more unhappy with this result than if the opposite decision were made on every issue (known as \emph{Anscombe's Paradox} \cite{anscombe1976frustration}). How can this arise? Consider a setting with 5 voters and 3 independent binary issues. The following table illustrates the preferences of each voter on each of the 3 issues: $+1$ is in favor and $-1$ is against:%
\begin{center}
    \begin{tabular}{ c | c c c }
      & Issue 1 & Issue 2 & Issue 3\\
     \hline
     $v_1 $ & \cellcolor{green!15}\texttt{+1} & \cellcolor{red!15}\texttt{-1} & \cellcolor{red!15}\texttt{-1}\\ 
     $v_2$ & \cellcolor{red!15}\texttt{-1} & \cellcolor{green!15}\texttt{+1} & \cellcolor{red!15}\texttt{-1}\\ 
     $v_3$ & \cellcolor{red!15}\texttt{-1} & \cellcolor{red!15}\texttt{-1} & \cellcolor{green!15}\texttt{+1}\\
     $v_4$ & \cellcolor{green!15}\texttt{+1} & \cellcolor{green!15}\texttt{+1} & \cellcolor{green!15}\texttt{+1}\\
     $v_5$ & \cellcolor{green!15}\texttt{+1} & \cellcolor{green!15}\texttt{+1} & \cellcolor{green!15}\texttt{+1}
    \end{tabular}
\end{center}

Now, assume each voter would only vote in favor of proposals that they agree with on more than half of the issues (in the paper, a voter will abstain when agreeing with a proposal on exactly half of the issues). 
Taking the majority on each topic yields the proposal $(+1, +1, +1)$. However, voters 1, 2, and 3 all disagree with a majority of this proposal. Therefore, if we posed this proposal for a vote, a majority of voters would vote against it. If, instead, we posed the opposite proposal $(-1, -1, -1)$, then voters 1, 2, and 3 would support it, and it would win the majority vote. Hence, in this scenario, the proposal comprising the \emph{minority} opinion on each topic wins the majority vote, whereas the proposal comprising the majority opinions fails to get majority support.

An equivalent view on the previous 
scenario
positions Anscombe's paradox in a broader context: instead of assuming a vote on a single proposal with people voting for/against it, let us assume that the vote happens between two competing proposals $p$ and $p'$ and each voter votes for whichever of $p$ and $p'$ agrees with their views on more topics, abstaining in case of equality.
Seen as such, Anscombe's paradox is the situation where an \emph{issue-wise majority} (IWM) proposal $p$ loses the majority vote against $p' = \overline{p}$, defined as the opposite proposal of $p$. A less extreme variant of the paradox, known as Ostrogorski's paradox \cite{rae1976ostrogorski} happens when an IWM proposal $p$ loses against some
proposal $p'$, not necessarily $\overline{p}$. Settling on the IWM proposal in such cases can lead to daunting situations where one of its opposers calls a final vote between $p$ and $p'$ that ``surprisingly'' unveils general dissatisfaction with what was otherwise a perfectly democratically chosen outcome.

Consequently, 
multi-issue aggregation mechanisms need to balance the tension between two majoritarian processes: majority on the individual topics and majority when proposals are compared to one another. 
In terms of the first, the chosen proposal should ideally stay somewhat close to IWM.  In terms of the second, the chosen proposal should not be easily refuted by calling a vote against some other proposal. 

Even when voters consider the issues to be of equal importance in their decision-making, we get paradoxical situations. 
However, in reality, voters rarely consider all issues to be equally important and often disagree on their importance; e.g., a Pew Research study from June 2023 indicated that in the United States, there were massive differences in perceived issue importance along partisan lines~\cite{pewpoll}. Some voting advice applications already attempt to account for personalized issue-importance, such as Smartvote~\cite{benesch2023voting}. 
Data from these applications can not only help assess how the current parties are aligning with the populace~\cite{mlsmartvote}, it can also suggest potential new party platforms. Such pre-existing infrastructure to get data on both voter opinions and issue importance underscores the pertinence of issue weights to modeling this problem setting.

\subsection{Our Contribution}

We study the aggregation of opinions on multiple independent binary issues with respect to two measures of majoritarianism: agreement with issue-wise majority and success in pairwise proposal comparisons. Our analysis considers  
two weighting models: \emph{external weights} and \emph{internal weights}. In the former, the policy-maker sets a weight to each issue reflecting its relative importance, and voters use weighted agreement when comparing any two proposals. The latter is the same, but each voter is free to choose their own weighting vector. We use the ``\emph{unweighted setting}'' to refer to the edge case where issues are equally-weighted.

\subsubsection{Condorcet winners.}

In the first part of the paper, we focus on the complexity of determining a \emph{Condorcet-winning proposal}: a proposal that does not lose in a direct vote against any other proposal. Under external weights, we find that any Condorcet winner has to be an IWM proposal, while this does not extend to internal weights. However, even in the unweighted setting with an odd number of voters, where the IWM proposal is unambiguous, checking whether this proposal is a Condorcet winner is co-NP-hard (answering an open question in \cite{constantinescu2023computing}).

\textbf{An Ostrogorski-free domain.} To mitigate this hardness result, it would be appealing to identify a large set of instances for which IWM proposals are Condorcet winners (i.e., Ostrogorski's paradox does not occur). If membership to this set could also be efficiently verified, this would allow for practically certifying ``safe instances'' where issue-wise majority is the right choice. We achieve this by the \emph{single-switch} condition of Laffond and Lain{\'e} \cite{laffond2006single}: a preference matrix over $\pm1$ is single-switch if it admits a \emph{single-switch presentation} --- a way to permute and potentially negate some columns such that $+1$ entries on each row form a prefix or a suffix. 
They show that for the unweighted case, this condition implies that Ostrogorski's paradox does not occur. We extend and simplify their analysis to show that the same holds under external weights (but not always for internal weights). We then provide a linear-time algorithm for checking whether the preference matrix is single-switch and prove that no-instances admit short proofs of this fact in the form of small forbidden subinstances (that can also be identified in linear time by a black-box reduction to the recognition problem which we have not encountered before). 

\textbf{Secondary implications.} Along the way, in this part, we make multiple secondary contributions: (i) we uncover an interesting topological connection: the set of single-switch presentations of a single-switch matrix can be compactly represented as the union of two mirror-image Möbius strips; (ii) our recognition algorithm for single-switch matrices proceeds by reducing to checking whether the columns of a matrix can be permuted so that the ones on each row form a prefix or a suffix --- while a linear-time algorithm is known for this \cite{elkind_lackner_dichotomous},\footnote{Under the name of recognizing \emph{voter/candidate-extremal-interval} preferences.} it relies on rather complex machinery --- we instead give a much simpler direct algorithm with the same guarantees; (iii) our simpler algorithm can be adapted to yield the simplest and at the same time most efficient algorithm for checking the \emph{single-crossing condition} in ranked social choice \cite{survey_restricted}. Similarly to the single-switch condition, the latter also admits a characterization in terms of small forbidden substructures \cite{bredereck_sc_forbidden_minors}, and finding such forbidden substructures can be achieved within the same time complexity using our black-box technique, a result which to the best of our knowledge is new.

\subsubsection{Representative majority-supported proposals.} Settling on a Condorcet-winning proposal would be ideal, especially under external weights where such proposals are by default IWM proposals, but in the absence of Condorcet winners, a compromise is needed. In fact, the hardness of checking whether an IWM proposal is a Condorcet winner can be seen positively: it is computationally demanding to find the proposal that defeats it, so we need not fear a vote being called against the defeating proposal. Hence, it is reasonable to relax the demanding Condorcet condition: the chosen proposal should, at the least, not lose against its opposite --- or, in the language of our first formulation of Anscombe's paradox above, should garner majority support. In the second part of the paper, we explore existence guarantees for majority-supported proposals that are as close as possible to an IWM proposal $p_{IWM}$.
So far, this has been studied in the unweighted model \cite{fritsch2022price,constantinescu2023computing}: a weakly majority-supported proposal agreeing in strictly more than half of the issues with $p_{IWM}$ exists and can be found in polynomial time, while achieving better guarantees is NP-hard. The word ``weakly'' can be dropped if majority is strict/unambiguous on at least one issue, i.e., some column of the preference matrix has differing numbers of $+1$'s and $-1$'s. We will be interested in the more complex weighted case.

\textbf{External weights.} 
We provide a matching guarantee to the unweighted case, showing that there always exists a weakly majority-supported proposal with strictly more than half the total weight in topics agreeing with $p_{IWM}$. Under a simple condition on certain higher-weight issues, we can also drop the word ``weakly.''

\textbf{Internal weights.} In sharp contrast, with as few as two different weight vectors, we construct families of instances where the distance between every weakly majority-supported proposal and the unique IWM proposal gets arbitrarily large. The severity of such examples is controlled by the maximum average topic weight $\tilde{w}_{max}$: we give a simple bound derived from a partition-based approach that is tight on a large portion of the range $\tilde{w}_{max} \in (0,1)$.

\textbf{More paradox-free instances.} Finally, we generalize Wagner's Rule of Three-Fourths \cite{wagner1983anscombe} for both external and internal weights: if the average weighted majority on the issues is at least $\frac{3}{4}$, then Anscombe's Paradox cannot occur. Without loss of generality, if $+1$ is a majority opinion on each topic, this translates to the total weight of $+1$'s in the preference matrix being at least $\frac{3}{4}$ of the total weight. A stronger condition precludes Ostrogorski's paradox under external weights: if on each column the relative weight of $+1$'s is at least $\frac{3}{4}$ of that column's total weight. This surprisingly simple check is a counterpart to the single-switch condition, once again giving a convenient characterization for a whole class of instances in which returning an IWM proposal is always a good choice.

\subsection{Further Related Work}\label{section:related}

Variations on the question of how best to reach consensus on a series of issues have been studied thoroughly. We first go over models where all topics are considered equally important.

Approval voting is a popular mechanism that is frequently used for single-winner and multi-winner elections alike~\cite{brams1978approval,fishburn1981analysis}. Here, each participant indicates their approval for a subset of candidates.
In contrast to our setting, not expressing the approval of a candidate does not give the same signal as voting for the ``no'' stance on an issue (which is a vote for the logical negation of the issue)~\cite{lackner2022approval}.

Another related field of study is judgment aggregation, where a series of judges have viewpoints on multiple topics, but there is external logical consistency required between the topics~\cite{list2012theory}. As in our problem, a reasonable method of reaching consensus is to take the majority opinion on each topic. However, the outcome may fail to be logically consistent --- this is the \emph{Discursive Dilemma}, and can occur with as few as 3 judges and 3 topics~\cite{kornhauser1993one}. There has been some investigation into conditions that avoid this paradox, like List's \emph{unidimensional alignment} \cite{list2003possibility},\footnote{The unidimensional alignment condition might appear to closely resemble the single-switch condition, as it essentially requires that the transposed preference matrix be single-switch. However, this is not equivalent, as rows and columns play different roles --- issue-wise majority aggregates along columns, not rows.} and other similar paradoxes under the name of \emph{compound majority paradoxes}~\cite{nurmi_compound_paradoxes}.

Our problem can also be viewed as a special instance of voting in combinatorial domains: multiple referenda with separable topics~\cite{brandt2016handbook}.
Multiple works explored generalizations of Anscombe's paradox 
and gave further impossibility results~\cite{benoit2010only,grandi2014common}, e.g., relating to the Pareto optimality of aggregation rules 
\cite{ozkal2006ensuring}.

Significant work has also been done to characterize when such paradoxes \emph{cannot} occur. 
Wagner proposed the Rule of Three-Fourths \cite{wagner1983anscombe}, preventing Anscombe's paradox,
as well as a generalization \cite{wagner1984avoiding}.
Laffond and Lainé showed that if no two voters disagree on too many issues, then Anscombe's is prevented \cite{laffond2013unanimity}, and for \emph{single-switch} preferences, Ostrogorski's does not occur~\cite{laffond2006single}.

We now survey proposals to augment various voting systems with weights, allowing voters to express their degrees of interest or investment in the topics. Storable voting allows participants to delay using their vote in a given election, and accumulate votes to use in later elections that they have more stake in~\cite{casella2005storable}. Quadratic voting proposes a somewhat similar system in which people are given an allotment of vote credits, and before a given election can buy a certain number of votes~\cite{lalley2018quadratic}. Both of these systems maintain that voters will use more votes for elections in which they feel strongly and believe they are likely to be pivotal in. 
Uckelman introduces a framework using goalbases to express cardinal (numeric) preferences over a combinatorial voting domain~\cite{uckelman2009more}. This, however, loses information by abstracting away the separability of issues: for us, the cardinal preferences are induced by the weighted Hamming distance.
Lang also considers augmenting combinatorial voting with preference weights and provides several computational complexity results \cite{lang2004logical}. Satisfaction approval voting~\cite{brams2015satisfaction} modifies approval voting by spreading a voter's total weight equally over all of the candidates they approve of. 
Finally, there is recent interest in studying how voters have varying stakes in elections and how to accommodate these stakes to limit distortion~\cite{brighouse2008democracy,flaniganaccounting}. 

\section{Model and Notation}\label{section:model}
For any non-negative integer $m$, write $[m] := \{1, \ldots, m\}$. Given a real number $x$, write $\sgn{x} \in \{-1, 0, 1\}$ for its \emph{sign}.
Note that for any two reals $x, y$, we have that $\sgn{x\cdot y} = \sgn{x}\cdot \sgn{y}$.
 
We consider a setting with $n$ voters and $t$ independent, binary issues/topics. The decision space for each issue is $\Bset := \{\pm 1\}$.
Each voter $i \in [n]$ is modeled as a dimension-$t$ vector $v_i \in \Bset^t$ indicating for each issue $j \in [t]$ the opinion/preference $v_{i, j} \in \Bset$ of voter $i$ on issue $j$. We call the matrix $\prefprof = (v_{i,j})_{i \in [n], j \in [t]}$ the \emph{preference profile}. We also write $\prefprof = (c_1, \ldots, c_t)$, where $c_1, \ldots, c_t \in \Bset^n$ are the columns of the matrix.

For each issue $j \in [t]$, we are consistent with previous literature ~\cite{wagner1983anscombe, wagner1984avoiding, fritsch2022price, constantinescu2023computing} and define the \emph{majority} $m_j \in [0, 1]$ on issue $j$ to be the fraction of voters that prefer $+1$ on it; i.e., the number of $+1$'s in $c_j$, divided by $n$. If $m_j > 0.5$, then the \emph{majority opinion} on issue $j$ is $+1$; if $m_j < 0.5$, then it is $-1$, and if $m_j = 0.5$, then both $+1$ and $-1$ are majority opinions on issue $j$. Equivalently, if we write $b_j$ for the sum of the entries in $c_j$ (i.e., the column's $\pm 1$-\emph{balance}), a majority opinion on issue $j$ is any $o \in \Bset$ satisfying $b_j \cdot o \geq 0$.

A \emph{proposal} is a vector $p \in \Bset^t$ that consists of a decision for each issue. We write $\overline{p}$ for the \emph{complement} of proposal $p$, which simply flips each bit of $p$; i.e., $\overline{p} = -p$. An \emph{issue-wise majority (IWM)} is a proposal $p$ where the decision on each topic is a majority opinion for the topic. 

We study two weighting models: \emph{external weights} and \emph{internal weights}. In the former, an externally supplied vector of non-negative weights $w = (w_1, \ldots, w_t)$ summing up to 1 is available, denoting the importance of each issue as seen collectively by the voters. The internal weights model generalizes this by having each voter $i \in [n]$ report an individual vector of weights $w_i = (w_{i, 1}, \ldots, w_{i, t})$; i.e., there need no longer be consensus on the importance of any fixed issue. For internal weights, we write $W$ for the matrix with rows $w_1, \ldots, w_n$. We call the \emph{voting instance} the pair $\votinginstance = (\prefprof, W)$ for internal weights and $\votinginstance = (\prefprof, w)$ for external weights.
We will also talk about the \emph{unweighted} model, which is simply external weights with $w = (1/t, \ldots, 1/t)$, and directly write $\votinginstance = \prefprof$ for it. For the remainder of this section, we assume external weights --- the internal weights model requires substantial additional notation so we postpone it to later on.

For any positive integer $m$, given two vectors $u, v \in \Bset^m$ and a vector of weights $w \in [0, 1]^m$ with unit sum, we write $\hamdist(u, v, w) := \sum_{j = 1}^{m} w_j \cdot \mathbb{I}(u_j \neq v_j)$ for the \emph{$w$-weighted Hamming distance} between $u$ and $v$. We omit the $w$ argument when referring to the unweighted Hamming distance. For convenience, we write $\innerprod{u}{v}_w := \sum_{j = 1}^{m} w_j \cdot u_j \cdot v_j$ for the standard \emph{$w$-weighted inner/dot-product}. One can easily show that $\innerprod{u}{v}_w = 1 - 2 \cdot \hamdist(u, v, w)$. 

Fix an instance $\votinginstance = (\prefprof, w)$ in the external weights model. For each voter $i$ with vote $v_i$ we define their \emph{individual} preference relation $\succcurlyeq_i$ between proposals. In particular, given two proposals $p, p' \in \Bset^t$, voter $i$ weakly prefers $p$ over $p'$, written $p \succcurlyeq_i p'$, iff $\hamdist(v_i, p, w) \leq \hamdist(v_i, p', w)$. Note that this is equivalent to $\innerprod{v_i}{p}_w \geq \innerprod{v_i}{p'}_w \iff \innerprod{v_i}{p - p'}_w \geq 0$. We write $\succ_i$ and $\approx_i$ for the strict and symmetric parts of $\succcurlyeq_i$, respectively. 
We define the \emph{collective} preference relation $\succcurlyeq_\votinginstance$ between proposals: given two proposals $p, p' \in \Bset^t$, the voters collectively weakly prefer $p$ over $p'$, written $p \succcurlyeq_\votinginstance p'$, iff $|\{i \in [n] \colon p \succ_i p'\}| \geq |\{i \in [n] \colon p' \succ_i p\}|$. Note that this is equivalent to $\sum_{i = 1}^{n}\sgn{\innerprod{v_i}{p - p'}_w} \geq 0$. We write $\succ_\votinginstance$ and $\approx_\votinginstance$ for the strict and symmetric parts of $\succcurlyeq_\votinginstance$, respectively. A proposal $p \in \Bset^t$ is a \emph{Condorcet winner} if for any other proposal $p' \in \Bset^t$ we have $p \succcurlyeq_\votinginstance p'$.

For a voting instance $\votinginstance$, Ostrogorski's paradox occurs if some IWM proposal $p_\textit{IWM}$ is not a Condorcet winner, Anscombe's paradox occurs if for some IWM proposal $p_\textit{IWM}$ we have $\overline{p_\textit{IWM}} \succ_\votinginstance p_\textit{IWM}$, 
and the Condorcet paradox happens if there is no Condorcet-winning proposal.

\section{Complexity of Determining a Condorcet Winner}

In this section, we prove that it is co-NP-hard to determine whether an instance $\votinginstance$ admits a Condorcet-winning proposal, even in the unweighted setting with odd $n$:

\begin{restatable}{theorem}{condorcet}\label{th:checking-condorcet-co-np-hard} Deciding whether an instance $\votinginstance = \prefprof$ admits a Condorcet winner is co-NP-hard in the unweighted setting with odd $n$.
\end{restatable}

This could be surprising given the following observation of \cite{laffond2006single} for the unweighted model, which we extend to external weights:

\begin{lemma}\label{lemma:condorcet-implies-iwm} Consider an external-weights instance $\votinginstance$ such that $p \in \Bset^t$ is a Condorcet winner for $\votinginstance$. Then, $p$ is an IWM for $\votinginstance$.
\end{lemma}
\begin{proof} Assume the contrary, then there is an issue $j \in [t]$ such that $p_j \cdot b_j < 0$. Consider the proposal $p^*$ obtained from $p$ by flipping $p_j$. Then, $p^* \succ_\votinginstance p,$ a contradiction. 
\end{proof}

\cref{lemma:condorcet-implies-iwm} shows that one can restrict the search space for Condorcet winners to IWM proposals. In the unweighted setting with odd $n$, there is a single such proposal, which we can assume without loss of generality to be $\mathbf{1} \in \Bset^t$. Nevertheless, even under these conditions, we will show that checking whether $\mathbf{1}$ is a Condorcet winner is co-NP-hard, or, equivalently, checking whether $\mathbf{1}$ is \emph{not} a Condorcet winner is NP-hard. The latter occurs if and only if there is a proposal $p \in \Bset^t$ such that $p \succ_\votinginstance \mathbf{1}$, which, recall, means that strictly more voters $i \in [n]$ prefer $p \succ_i \mathbf{1}$ than $\mathbf{1} \succ_i p$. Hence, it suffices to prove that the following problem is NP-hard:

\begin{tcolorbox}[boxsep=0mm]
\textbf{Problem ``\fauxsc{Major}''}\\
\textbf{Input}: Instance $\votinginstance = \prefprof$ in the unweighted setting with odd $n$ such that $\textbf{1}$ is the issue-wise majority. \\
\textbf{Output}: Does there exist a proposal $p \in \Bset^t$ s.t.~$p \succ_\votinginstance \mathbf{1}$?
\end{tcolorbox}

To show its hardness, we need the following auxiliary problem:

\begin{tcolorbox}[boxsep=0mm]
\textbf{Problem ``\fauxsc{Unanim}''}\\
\textbf{Input}: Voting instance $\votinginstance = \prefprof$ in the unweighted setting. \\
\textbf{Output}: Does there exist a proposal $p \in \Bset^t$ s.t.~$p \succ_i \mathbf{1}$ for all $i \in [n]$ (to be read ``$p$ \emph{unanimously} defeats $\mathbf{1}$'')?
\end{tcolorbox}

\fauxsc{Unanim} is NP-hard \cite[{Theorem 2}]{complexity_deliberative_coalition}, but the proof in \cite{complexity_deliberative_coalition} is relatively complicated: we give a simpler one in \cref{app:condorcet} by noting the equivalence to choosing a subset of columns of $\prefprof$ that
sum up to a negative amount on each row (we also give a similar reformulation of \fauxsc{Major} for the interested reader).

\begin{lemma}\label{lemma:major-is-np-hard}
\fauxsc{Major} is NP-hard.
\end{lemma}

\begin{proof} We reduce from the NP-hard problem \fauxsc{Unanim}. Consider an instance $\votinginstance = \prefprof$ of \fauxsc{Unanim} with $n$ voters. If there is an issue $j \in [t]$ disapproved by all voters in $\prefprof$, then $\prefprof$ is a yes-instance of \fauxsc{Unanim}: all voters prefer the proposal with $+1$ in all coordinates except the $j$-th to proposal $\mathbf{1}$. This case can be easily detected in polynomial time, so we henceforth assume the contrary.

We build an instance $\votinginstance' = \prefprof'$ of \fauxsc{Major} from $\prefprof$ by adding $n - 1$ voters approving all issues.
For $\prefprof'$ to be a valid instance for \fauxsc{Major} we need that $2n - 1$ is odd (which it is) and that $\mathbf{1}$ is the issue-wise majority. The latter holds because at least $n - 1 + 1 = n$ voters approve of each issue: the $n - 1$ added ones and at least one from the first $n$ by our assumption.
It remains to show that a proposal $p \in \Bset^t$ unanimously defeats $\mathbf{1}$ in $\prefprof$ iff it majority-defeats $\mathbf{1}$ in $\prefprof'$.

Assume $p \in \Bset^t$ unanimously defeats $\mathbf{1}$ in $\prefprof$. Then, each of the first $n$ voters in $\prefprof'$ prefers $p$ to $\mathbf{1}$. Since there are only $n - 1 < n$ other voters in $\prefprof'$, a majority of the voters in $\prefprof'$ prefer $p$ to $\mathbf{1}$.

Conversely, assume $p \in \Bset^t$ majority-defeats $\mathbf{1}$ in $\prefprof'$. Clearly, $p \neq \mathbf{1}$ has to hold, so all of the $n - 1$ added voters prefer $\mathbf{1}$ to $p$. To counteract this, since $p \succ_{\votinginstance'} \mathbf{1},$ the first $n$ voters in $\prefprof'$ must prefer $p$ to $\mathbf{1}$, meaning that $p$ unanimously defeats $\mathbf{1}$ in $\prefprof$. 
\end{proof}

For completeness, we put the pieces together to give a self-contained proof of \cref{th:checking-condorcet-co-np-hard} in \cref{app:condorcet}.

\section{An Ostrogorski-free Domain}

As we have seen, at least for external weights, a Condorcet-winning proposal has to be an issue-wise majority proposal. Yet, we proved that determining whether one of them is actually Condorcet-winning is co-NP-hard, even in the unweighted case with odd $n$, where there is only one such proposal to check. To mitigate this hardness result, it would be useful if we could identify a large set of instances for which IWM proposals are guaranteed to be Condorcet-winning, i.e., Ostrogorski's paradox does not occur. Laffond and Lain{\'e} \cite{laffond2006single} introduced the \emph{single-switch} condition, which achieves exactly this goal for the unweighted setting. Furthermore, they showed that it is the most general condition preventing Ostrogorski's paradox among conditions that do not consider the multiplicities of the votes (i.e., conditions defining a \emph{domain}) or whether a vote is negated or not (i.e., they only look at the set $\{\{v_i, \overline{v_i}\} \mid i \in [n]\}$ and not at how many times each $v_i$ or $\overline{v_i}$ is repeated). In particular, if an instance in the unweighted model is not single-switch, then it is possible to add copies of some of the votes $v_i$ (or their negations $\overline{v_i}$) so that some issue-wise majority proposal is not a Condorcet winner. Two important questions underpinning their condition are: (i) Does it still guarantee the existence of a Condorcet winner in the (at least externally) weighted setting? (ii) Is it possible to check whether it applies in polynomial time? If not, are there short proofs of this fact? In this section, we answer all these questions in the affirmative.

A preference profile (matrix) $\prefprof = (c_1, \ldots, c_t)$ is \emph{single-switch} (SSW) if we can flip (multiply by $-1$ all entries in) some columns and then permute the columns to get a new profile $\prefprof'$ such that $+1$ entries on every row form either a prefix or a suffix, in which case we say that $\prefprof'$ is an SSW presentation of $\prefprof$. We allow flipping no columns or leaving all columns in their original place. Intuitively, issues are arranged along a left-right axis. Left-wing voters approve a prefix of issues, with the length depending on their tolerance, while right-wing voters similarly approve a suffix of issues.\footnote{This shares similarities with several related concepts, such as single-peaked and single-crossing preferences. However, unlike most other notions, we allow issues to be flipped before ordering them, as they can be logically negated without changing meaning.} See \cref{fig:single-switch} for an illustration of the notion. A voting instance $\votinginstance$ is single-switch if its preference profile $\prefprof$ is single-switch.

\newcommand{\drawcell}[2]{\draw[line width = 2.0mm, opacity = 0.15, color = {#2}] ({#1}.west) -- ({#1}.east);}
\newcommand{\drawline}[3]{\draw[line width = 2.0mm, opacity = 0.15, color = {#3}] ({#1}.west) -- ({#2}.east);}
\newcommand{\drawweight}[2]{\draw[line width = 2.0mm, opacity = 0.15, color = {#2}, text = purple] ({#1}.west) -- ({#1}.east);}
\newcommand{\czero}{red}
\newcommand{\cone}{green}
\setlength{\tabcolsep}{4pt}

\begin{figure}[t]
\centering
\begin{subfigure}{0.45\linewidth}
\centering
\begin{tabular}{cccccc}
      1 & 2 & 3 & 4 & 5 & 6  \\
      \hline
       \tikzmarknode{ns_1_1}{\texttt{+1}} & \tikzmarknode{ns_1_2}{\texttt{+1}} & \tikzmarknode{ns_1_3}{\texttt{+1}} & \tikzmarknode{ns_1_4}{\texttt{-1}} & \tikzmarknode{ns_1_5}{\texttt{+1}} & \tikzmarknode{ns_1_6}{\texttt{-1}} \\
       \tikzmarknode{ns_2_1}{\texttt{+1}} & \tikzmarknode{ns_2_2}{\texttt{-1}} & \tikzmarknode{ns_2_3}{\texttt{+1}} & \tikzmarknode{ns_2_4}{\texttt{+1}} & \tikzmarknode{ns_2_5}{\texttt{-1}} & \tikzmarknode{ns_2_6}{\texttt{+1}} \\
       \tikzmarknode{ns_3_1}{\texttt{+1}} & \tikzmarknode{ns_3_2}{\texttt{+1}} & \tikzmarknode{ns_3_3}{\texttt{-1}} & \tikzmarknode{ns_3_4}{\texttt{-1}} & \tikzmarknode{ns_3_5}{\texttt{+1}} & \tikzmarknode{ns_3_6}{\texttt{-1}} \\
\end{tabular}
\caption{Profile $\prefprof$.}
\label{fig:ss_a}
\end{subfigure} 
\begin{subfigure}{0.53\linewidth}
\centering
\begin{tabular}{cccccc}
      2 & 1 & 3 & 4 & $\overline{5}$ & 6  \\
      \hline
       \tikzmarknode{ss_1_1}{\texttt{+1}} & \tikzmarknode{ss_1_2}{\texttt{+1}} & \tikzmarknode{ss_1_3}{\texttt{+1}} & \tikzmarknode{ss_1_4}{\texttt{-1}} & \tikzmarknode{ss_1_5}{\texttt{-1}} & \tikzmarknode{ss_1_6}{\texttt{-1}} \\
       \tikzmarknode{ss_2_1}{\texttt{-1}} & \tikzmarknode{ss_2_2}{\texttt{+1}} & \tikzmarknode{ss_2_3}{\texttt{+1}} & \tikzmarknode{ss_2_4}{\texttt{+1}} & \tikzmarknode{ss_2_5}{\texttt{+1}} & \tikzmarknode{ss_2_6}{\texttt{+1}} \\
       \tikzmarknode{ss_3_1}{\texttt{+1}} & \tikzmarknode{ss_3_2}{\texttt{+1}} & \tikzmarknode{ss_3_3}{\texttt{-1}} & \tikzmarknode{ss_3_4}{\texttt{-1}} & \tikzmarknode{ss_3_5}{\texttt{-1}} & \tikzmarknode{ss_3_6}{\texttt{-1}} \\
\end{tabular}%
\caption{Single-switch presentation of $\prefprof$.}
\label{fig:ss_b}
\end{subfigure}
\caption{The profile $\prefprof$ in \cref{fig:ss_a} is single-switch because its columns can be permuted and flipped as in \cref{fig:ss_b} to ensure that ones on each row form a prefix or a suffix.}
\label{fig:single-switch}
\begin{tikzpicture}[overlay,remember picture, shorten >=-3pt, shorten <= -3pt]
\drawcell{ns_1_1}{\cone}
\drawcell{ns_1_2}{\cone}
\drawcell{ns_1_3}{\cone}
\drawcell{ns_2_1}{\cone}
\drawcell{ns_2_2}{\czero}
\drawcell{ns_2_3}{\cone}
\drawcell{ns_2_4}{\cone}
\drawcell{ns_2_6}{\cone}
\drawcell{ns_3_1}{\cone}
\drawcell{ns_3_2}{\cone}
\drawcell{ns_1_4}{\czero}
\drawcell{ns_1_6}{\czero}
\drawcell{ns_3_3}{\czero}
\drawcell{ns_3_4}{\czero}
\drawcell{ns_3_6}{\czero}
\drawcell{ns_1_5}{\cone}
\drawcell{ns_2_5}{\czero}
\drawcell{ns_3_5}{\cone}
\drawline{ss_1_1}{ss_1_3}{\cone}
\drawline{ss_2_2}{ss_2_6}{\cone}
\drawline{ss_3_1}{ss_3_2}{\cone}
\drawline{ss_1_4}{ss_1_6}{\czero}
\drawcell{ss_2_1}{\czero}
\drawline{ss_3_3}{ss_3_6}{\czero}
\end{tikzpicture}
\end{figure}

\subsection{For External Weights Single-Switch Prevents Ostrogorski's Paradox}

We find that, assuming external-weights, the single-switch condition guarantees that all IWM proposals are Condorcet winners. To show this, we first show that every issue-wise majority proposal does not lose against its opposite, i.e., Anscombe's paradox does not occur. We do this by streamlining and adapting the argument in \cite{laffond2006single} (which was only for the unweighted model). Because the single-switch condition is closed under removing issues, the general statement then follows easily by noting that, under external weights, Ostrogorski's paradox happens if and only if there is a subset of issues inducing an instance where Anscombe's paradox happens. The details are deferred to \cref{app:external-weights-single-switch}.

\begin{restatable}{theorem}{externalswitchcondorcet}\label{th:external-single-switch-no-condorcet} In the external-weights model, every issue-wise majority proposal of a single-switch instance is a Condorcet winner.
\end{restatable}

\subsection{Recognizing Single-Switch Profiles}

The result in the previous section is particularly appealing: in the external-weights model, if the preferences are single-switch, any issue-wise majority proposal is a Condorcet winner. This bypasses our previous hardness result in the case of single-switch preferences. However, this is only useful provided one can quickly tell whether a given profile $\prefprof$ is single-switch or not. In this section, we show that this can be determined in linear time, i.e., $O(nt)$. For yes-instances, our algorithm also determines an SSW presentation $\prefprof'$ (implicitly also the permutation and flips used to obtain it). Given $\prefprof'$, we also characterize the set of all SSW presentations as the union of two ``orbits'' around $\prefprof'$ and its column-reversal. These orbits can be attractively interpreted topologically as two mirror-image Möbius strips. To begin, we need the following observation following easily from the case $n = 1$. See \cref{app:reco-ssw} for the proof.

\begin{restatable}{lemma}{lemmacyclicshift}\label{lemma:cyclic-shift-and-negate} Consider a profile $\prefprof$ admitting an SSW presentation $\prefprof' = (c_1, \ldots, c_t)$. Then, $\prefprof'_r := (c_2, \ldots, c_t, \overline{c_1})$ is also a SSW presentation of $\prefprof$. Furthermore, any $t$ (circularly) consecutive columns in $\prefprof'' := (c_1, \ldots, c_t, \overline{c_1}, \ldots, \overline{c_t})$ form an SSW presentation of $\prefprof$.
\end{restatable}

Hence, any SSW presentation $\prefprof'$ of a profile $\prefprof$ corresponds to a set of $2t$ such presentations that we call the \emph{orbit} $O_{\prefprof'}$ of $\prefprof'$. Formally, these are the $2t$ profiles that can be obtained by taking $t$ (circularly) consecutive columns in $\prefprof''$ in the above. 
Note that the orbits of any two SSW presentations either coincide or are disjoint, so the set of all orbits partitions the set of SSW presentations of $\prefprof$. Also, the $2t$ profiles in $O_{\prefprof'}$ are pairwise distinct, which can be easily seen by considering the case $n = 1$, under which $\prefprof''$ is circularly equivalent to a list of $t$ minus ones followed by $t$ ones. This reasoning additionally allows us to assign to each orbit a \emph{representative}, namely the profile with all $-1$'s on the first row: 

\begin{corollary}\label{coro:orbit-minus-one} Every orbit contains exactly one profile where the first row is all $-1$'s.
\end{corollary}

Orbits can be understood through a topological lens: For the orbit $O_{\prefprof'}$ of $\prefprof' = (c_1, \ldots, c_t)$ take an $n \times t$ rectangular piece of paper and write the columns $c_1, \ldots, c_t$ on the front and $\overline{c_1}, \ldots, \overline{c_t}$ on the back, such that for each $i \in [t]$, column $c_i$ on the front aligns with column $\overline{c_i}$ on the back. Then, give the paper a length-wise half-twist and glue the left and right sides to form a surface known as a Möbius strip: see~\cref{fig:moebius}. Cutting along the width of the strip between \emph{any} two columns recovers an $n \times t$ piece of paper with one SSW presentation on one side and its opposite on the other side. 
In high-level terms, each orbit is topologically a Möbius strip.

\begin{figure}[t]
    \centering
    \begin{tikzpicture}[scale=0.8]
        \fill[green!20] (9:2) arc[start angle=9, end angle=216-0.5, radius = 2] -- (216-0.5:3) arc[start angle=216-0.5, end angle=9, radius=3] .. controls (0:{3/cos(4.5)}) and (0:{2.5/cos(4.5)}) .. (0:2.5) (0:2.5) .. controls (0:{2.5/cos(4.5)}) and (0:{2/cos(4.5)}) .. (9:2);
        
        \fill[red!20] (360-9:2) arc[start angle=360-9, end angle=216+0.5, radius = 2] -- (216+0.5:3) arc[start angle=216+0.5, end angle=360-9, radius=3] .. controls (0:{3/cos(4.5)}) and (0:{2.5/cos(4.5)}) .. (0:2.5) (0:2.5) .. controls (0:{2.5/cos(4.5)}) and (0:{2/cos(4.5)}) .. (360-9:2);
    
        \draw[thick] (9:2) arc[start angle=9, end angle=216-0.5, radius = 2];
        \draw[thick] (216+0.5:2) arc[start angle=216+0.5, end angle=360-9, radius = 2];
        
        \draw[thick] (9:3) arc[start angle=9, end angle=216-0.5, radius = 3];
        \draw[thick] (216+0.5:3) arc[start angle=216+0.5, end angle=360-9, radius = 3];

        \draw[thick] (9:2) .. controls (0:{2/cos(4.5)}) and (0:{2.5/cos(4.5)}) .. (0:2.5);
        \draw[thick] (0:2.5) .. controls (0:{2.5/cos(4.5)}) and (0:{3/cos(4.5)}) .. (360-9:3);
        \draw[thick, dashed] (9:3) .. controls (0:{3/cos(4.5)}) and (0:{2.5/cos(4.5)}) .. (0:2.5);
        \draw[thick, dashed] (0:2.5) .. controls (0:{2.5/cos(4.5)}) and (0:{2/cos(4.5)}) .. (360-9:2);
        
        \foreach \i in {1,2,3,4,5,7,8,9} \draw[thick] (\i * 36:2) -- (\i * 36:3);
        \draw[thick] (216-0.5:2) -- (216-0.5:3);
        \draw[thick] (216+0.5:2) -- (216+0.5:3);
    
        \foreach [count=\i] \chara in {a,...,f}{
            \node[rotate={36 * 7-36 * \i + 22 - 36 - 90}] at (7 * 36 -\i * 36 - 36 + 18:2.5) {\Large${c_{\i}}$};
        };
        \foreach \i in {7,8,9,10}{
            \node[rotate={17 * 36 - 36 * \i + 14 - 36 + 90}] at (17 * 36 - \i * 36 - 36 + 18:2.5) 
            {\Large$\overline{c_{\i}}$};
        };
        
        \fill[fill=black!95, rounded corners=1.5pt] (216-1.6:1.94) -- (216-1.4:3.06) -- (216+1.4:3.06) -- (216+1.6:1.94) -- cycle;
    \end{tikzpicture}
    \caption{Möbius strip of orbit $O_{\prefprof'}$ for $\prefprof' = (c_1, \ldots, c_{10})$. We start with a rectangular piece of paper of length 10 and write $(c_1, \ldots, c_{10})$ on the (green) front side and $(\overline{c_1}, \ldots, \overline{c_{10}})$ on the (red) backside. We then give the paper a length-wise half-turn and glue the endpoints (bold strip). This gives raise to a surface with a single continuous side. 
    }
    \label{fig:moebius}
\end{figure}
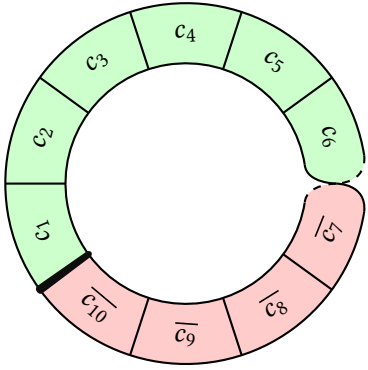

To check whether a profile $\prefprof$ is single-switch, by \cref{coro:orbit-minus-one}, it suffices to check for presentations with all $-1$'s in the first row: all other presentations are generated by the orbits of such presentations. There is a simple strategy to achieve this: flip columns in $\prefprof$ to make the first row all $-1$'s, and then check whether columns in the resulting profile can be permuted to ensure that ones on each row form a prefix or a suffix. This amounts to recognizing single-switch-no-flips profiles:
A profile $\prefprof$ is \emph{single-switch-no-flips} (SSWNF) if its columns can be permuted to get a new profile $\prefprof'$ such that $+1$ entries on every row form either a prefix or a suffix, in which case we say that $\prefprof'$ is an SSWNF presentation of $\prefprof$.

\textbf{Recognizing single-switch-no-flips profiles.} Telling whether a profile $\prefprof = (c_1, \ldots, c_t)$ is single-switch-no-flips can be achieved by appending a negated copy of $\prefprof$ underneath \cite{elkind_lackner_dichotomous} and running a solver for the  \emph{Consecutive Ones Problem} (C1P), which can be solved in $O(n t)$ time \cite{booth_lueker}, implying the same about our problem. However, such solvers are complicated and notoriously error-prone: most available implementations fail on at least some edge cases \cite{experimental_pq_trees}. Moreover, reducing to C1P does not utilize the additional structure present in our problem and hence does not shed light on the structure of all solutions, as we set out to do. We give a much simpler algorithm achieving the $O(n t)$ time-bound: Find an index $x$ maximizing $d_H(c_1, c_x)$. Then, sort (using Counting Sort) the columns based on their Hamming distance from $c_x$ to get a profile $\prefprof' = (c_1', \ldots, c_t')$ where $d_H(c_x, c_i') \leq d_H(c_x, c_{i + 1}')$ for $i \in [t - 1]$ (i.e., ties in Hamming distance can be broken arbitrarily). We claim that either $\prefprof'$ is the unique SSWNF presentation of $\prefprof$ (up to reversing the order of the columns), or there is no such presentation, so we can easily check in additional $O(n t)$ time whether the candidate solution works. All required claims shown formally in \cref{app:reco-ssw}:

\begin{restatable}{theorem}{recosswnf}\label{th:sswnf-reco} There is a \emph{simple} $O(nt)$ algorithm computing (or deciding the inexistence of) an SSWNF presentation of a profile $\prefprof$. Moreover, if it exists, this presentation is unique up to reversing column order.
\end{restatable}

\cref{app:reco-ssw} also provides a much ampler discussion of related work for this sub-problem, including the relation between our algorithm and previous algorithms for recognizing \emph{single-crossing preferences}. As a bonus, it gives a similar simpler, more efficient algorithm for recognizing single-crossing preferences, running in time $O(nt\sqrt{\log{n}})$, improving state of the art \cite[Algorithm 4]{survey_restricted}.

\textbf{Putting it together.} To decide whether a profile $\prefprof$ is single-switch, we flip columns in $\prefprof$ to get a profile $\prefprof'$ with only $-1$'s in the first row and then use the algorithm in \cref{th:sswnf-reco} to find an SSWNF presentation $\prefprof''$ of $\prefprof'$ (and hence also $\prefprof$). If it exists, this presentation is unique up to column reversal, so we can also characterize the set of all SSW presentations of $\prefprof$ by unioning the orbits of $\prefprof''$ and its column-reversal. Note that these two orbits may coincide for pathological input profiles $\prefprof$. 

\begin{theorem} There is an $O(nt)$ algorithm computing (or deciding the inexistence of) an SSW presentation of a profile $\prefprof$. If the algorithm returns a presentation $\prefprof''$, let $\prefprof''_r$  be $\prefprof''$ with the order of the columns reversed, then the set of all SSW presentations of $\prefprof$ is $O_{\prefprof''} \cup O_{\prefprof''_r}$.
\end{theorem}

\subsection{Forbidden Subprofiles Characterization of Single-Switch Preferences}

Whenever the single-switch condition is not satisfied, it would be useful if there were a short proof of this fact: a small subprofile that is not single-switch. Formally, a profile/matrix $\prefprof$ contains a profile/matrix $\prefprof'$ as a \emph{subprofile}/\emph{submatrix} if we can remove (possibly zero) rows and columns from $\prefprof$ to get $\prefprof'$ up to permuting rows and columns. Note that existence is not immediate: 
there could exist arbitrarily large matrices not satisfying the condition but all of whose proper submatrices do.
We show that this is not the case: either the condition holds, or there is a $3 \times 4$ or $4 \times 3$ submatrix witnessing that this is not the case, as in the following:

\begin{restatable}{theorem}{forbiddensswmain}\label{th:forbidden-ssw-main} A profile $\prefprof$ is single-switch if and only if it does not contain as a subprofile $\prefprof_1^a, \prefprof_2^a$ and any profile that can be obtained from them by flipping rows and columns: 
\[
\prefprof_1^a =\begin{bmatrix}
\tikzmarknode{p1a11}{\texttt{-1}} & \tikzmarknode{p1a12}{\texttt{-1}} & \tikzmarknode{p1a13}{\texttt{-1}} & \tikzmarknode{p1a14}{\texttt{-1}} \\
\tikzmarknode{p1a21}{\texttt{+1}} & \tikzmarknode{p1a22}{\texttt{+1}} & \tikzmarknode{p1a23}{\texttt{-1}} & \tikzmarknode{p1a24}{\texttt{-1}} \\
\tikzmarknode{p1a31}{\texttt{+1}} & \tikzmarknode{p1a32}{\texttt{-1}} & \tikzmarknode{p1a33}{\texttt{+1}} & \tikzmarknode{p1a34}{\texttt{-1}}
\end{bmatrix}\quad
\prefprof_2^a = \begin{bmatrix}
\tikzmarknode{p2a11}{\texttt{-1}} & \tikzmarknode{p2a12}{\texttt{-1}} & \tikzmarknode{p2a13}{\texttt{-1}} \\
\tikzmarknode{p2a21}{\texttt{+1}} & \tikzmarknode{p2a22}{\texttt{-1}} & \tikzmarknode{p2a23}{\texttt{-1}} \\
\tikzmarknode{p2a31}{\texttt{-1}} & \tikzmarknode{p2a32}{\texttt{+1}} & \tikzmarknode{p2a33}{\texttt{-1}} \\
\tikzmarknode{p2a41}{\texttt{-1}} & \tikzmarknode{p2a42}{\texttt{-1}} & \tikzmarknode{p2a43}{\texttt{+1}} 
\end{bmatrix}\quad
\]%
\begin{tikzpicture}[overlay,remember picture, shorten >=-3pt, shorten <= -3pt]
\drawcell{p1a11}{\czero}
\drawcell{p1a12}{\czero}
\drawcell{p1a13}{\czero}
\drawcell{p1a14}{\czero}
\drawcell{p1a21}{\cone}
\drawcell{p1a22}{\cone}
\drawcell{p1a23}{\czero}
\drawcell{p1a24}{\czero}
\drawcell{p1a31}{\cone}
\drawcell{p1a32}{\czero}
\drawcell{p1a33}{\cone}
\drawcell{p1a34}{\czero}
\drawcell{p2a11}{\czero}
\drawcell{p2a12}{\czero}
\drawcell{p2a13}{\czero}
\drawcell{p2a21}{\cone}
\drawcell{p2a22}{\czero}
\drawcell{p2a23}{\czero}
\drawcell{p2a31}{\czero}
\drawcell{p2a32}{\cone}
\drawcell{p2a33}{\czero}
\drawcell{p2a41}{\czero}
\drawcell{p2a42}{\czero}
\drawcell{p2a43}{\cone}
\end{tikzpicture}%
\end{restatable}%

We prove \cref{th:forbidden-ssw-main} in \cref{app:forbidden} by combining a similar characterization for single-switch-no-flips profiles given in \cite{Terzopoulou_Karpov_Obraztsova_2021} (under the name \emph{voter/candidate-extremal-interval} preferences) with our insight that to go to the no-flips version it suffices to make one row all $-1$'s. Henceforth, we call the $3 \times 4$ and $4 \times 3$ preference profiles in the theorem above \emph{forbidden subprofiles}. Then, the theorem says that $\prefprof$ is single-switch if and only if it contains no forbidden subprofiles. Note how this implies that single-switch profiles are relatively rare: the probability that a random binary $n \times t$ matrix is single-switch tends to zero as $n$ and $t$ tend to infinity.

\textbf{Finding forbidden subprofiles.} So far, we have seen that non-membership to the class of single-switch preferences admits short proofs, but can such proofs also be constructed efficiently? Given some no-instance, it is straightforward to determine which forbidden subprofiles occur in it in time $O(n^3t^4 + n^4 t^3)$.
In contrast, our recognition algorithm runs in time $O(nt)$, but does not identify a forbidden subprofile. We will now assume our $O(nt)$ recognition algorithm as a black box and show how to identify a forbidden subprofile for a given no-instance $\prefprof$ in time $O(nt)$.

Let us first describe an $O(n^2t + nt^2)$ approach: one at a time, try to remove each row and each column of $\prefprof$, i.e., $n + t$ removal attempts; if doing so makes the resulting profile a yes-instance, undo the removal, and otherwise let it persist. At the end, the ensuing no-instance $\prefprof'$ is a subprofile of $\prefprof$ whose proper subprofiles are yes-instances, so $\prefprof'$ is a forbidden subprofile, completing the argument.

We now modify the previous idea to run in time $O(nt)$ by removing multiple rows/columns at a time. We will first only remove rows, and then, starting from the resulting profile, only columns. The reasoning for columns is entirely analogous, so we only describe the procedure for rows: partition the rows into 5 groups $G_1, \ldots, G_5$, each of size roughly $n / 5$. Because all forbidden subprofiles are of size $3 \times 4$ or $4 \times 3$, any occurrence of a forbidden subprofile in $\prefprof$ only uses rows from at most 4 of the 5 groups. Consequently, we can find a group $G_i$ such that removing all rows in $G_i$ from $\prefprof$ keeps the property that $\prefprof$ is a no-instance. Doing so requires at most 5 calls to the recognition algorithm, so it can be done in overall time $O(nt)$. Ignoring for brevity the cases where $n$ is not divisible by 5, this reasoning shows how to reduce $n$ to $4n/5$ in time $O(nt)$. Applying the same reasoning iteratively until $n$ goes below 5 takes total time $O(nt)$ because the geometric series $\sum_{i = 0}^\infty (4/5)^i$ converges. 

\begin{theorem}\label{th:fast-forbidden-subprofiles-ssw} Given a non-single-switch profile $\prefprof$, a forbidden subprofile of $\prefprof$ can be determined in time $O(nt)$.
\end{theorem}

We note that the previous idea applies more broadly; e.g., for single-crossing preferences, which admit a characterization in terms of two small \emph{forbidden subinstances} \cite{bredereck_sc_forbidden_minors}, our $O(nt\sqrt{\log{n}})$ recognition algorithm can be bootstrapped to also produce a forbidden subinstance for no-instances within the same time bound. A formal statement and more details can be found in \cref{app:forbidden}.

\section{Anscombe's Paradox}\label{sec:anscombe}
When preferences are not single-switch, determining whether an IWM proposal is a Condorcet winner is co-NP hard. In light of this, we focus on the most diabolical subset of Ostrogorski paradox instances: those inducing Anscombe's paradox (where an IWM proposal is defeated by its complement, or, equivalently, an IWM proposal fails to get majority support). If Anscombe's paradox occurs, a natural question is: ``How close can we get to any given IWM while still requiring that the proposal gets majority support?'' 

We first explore this question under external weights, i.e., in instances $\mathcal{I} = (\mathcal{P}, w)$ where all voters share the same, unit-sum weights vector $w$. Then, we introduce the necessary notation and study it for internal weights. Finally, we give a simple characterization of a broad swath of instances that avoid Anscombe's paradox entirely for internal weights. We assume throughout that $t > 1$, as Anscombe's paradox does not occur with one topic, and without loss of generality that $m_j \geq 0.5$ for all $j \in [t]$ (i.e., that +1 is a majority opinion on all topics).

Formally, some voter $i$ supports (approves of) a proposal $p$ if $\hamdist(v_i, p, w) < 1/2$, opposes (disapproves of) $p$ if $\hamdist(v_i, p, w) > 1/2$, and is indifferent to $p$ if $\hamdist(v_i, p, w) = 1/2$. A proposal is \emph{strictly majority-supported} if more people support it than oppose it and \emph{weakly majority-supported} if no more people oppose it than support it. Our definition of majority support matches 
\cite{constantinescu2023computing} but differs from \cite{fritsch2022price} (where indifferent voters count towards the proposal's support).

\subsection{External Weights}
In the unweighted case, it is straightforward to argue that for any IWM, there exists a weakly majority-supported proposal within distance $\leq \frac{1}{2} + \frac{1}{2t}$ because at least one proposal in every complementary pair $(p, \overline{p})$ gets weak majority support (and at least one pair satisfies the distance bound for both proposals).
A slightly better guarantee of distance $< \frac{1}{2}$ holds by a more difficult proof~\cite{constantinescu2023computing,fritsch2022price}. 
For external weights, the complementary pairs argument no longer gives a bound close to $\frac{1}{2}$ if no subset of topic weights sum up close to $\frac{1}{2}$. One may hope to reduce to the unweighted case by splitting topics into multiple equal-weight topics and use the $< \frac{1}{2}$ bound there, but the resulting majority-supported proposals may have different values for an original topic's clones, making it hard to translate to proposals in the original instance. Despite these setbacks, we surprisingly find that the $< \frac{1}{2}$ guarantee still holds for external weights. Our proof, deferred to~\cref{app:existence_external_iwm}, simplifies and adapts the argument in \cite{constantinescu2023computing}. 
We also guarantee \emph{strict} majority support if there is a strict majority in at least one \emph{relevant} topic, roughly meaning topics with high enough weight to be the tipping point in a vote (see \cref{app:relevant_topics} for a formal definition). 

\begin{restatable}{theorem}{externaliwm}\label{thm:external_iwm}
    For any $\mathcal{I} = (\mathcal{P}, w)$ and $p_{IWM}$, there is a weakly majority supported
    proposal $p$ with $d_{H}(p, p_{IWM}, w) < 1/2$. If majority is strict in any relevant topic, ``weak'' can be replaced with ``strict''. 
\end{restatable}

\subsection{Internal Weights}\label{subsec:anscombe_internal}
We now explore a model where individuals can have unique weight vectors, expressing not only diverse preferences on issue outcomes but also differing opinions on relative topic importance.

\textbf{Internal Weights Model.} 
In the internal weights model, an instance $\mathcal{I} = (\mathcal{P}, W)$ consists of a preference profile $\mathcal{P}$ and a weight profile $W$ with rows $w_1, \dots, w_n$ where each weight vector $w_i$ corresponds to voter $i$, is non-negative, and sums to 1. The \emph{average weight vector} is defined as $\tilde{w} := \frac{1}{n}\sum_{i=1}^n w_i$. Zero entries in the average weight vector correspond to issues that no voters placed any weight on (and hence can be ignored). We assume no such topics exist without loss of generality.
We define the \emph{majority} for a given topic $j$ to be $m_j := \frac{1}{n \avgweight_j}\sum_{i=1}^n w_{i,j} \cdot \mathbb{I}(v_{i,j} = +1)$. This is the fraction of voter weight placed on that issue that prefers $+1$. Note that this agrees with our previous definition for external weights (where it was just the fraction of voters that prefer $+1$ on that topic). 
The average majority for a given preference profile is defined as $\avgmaj := \sum_{j=1}^t \avgweight_j m_j$. This naturally weights consensus on issues proportionally to how important those issues are to the population.

Under external weights, we could give a constant upper bound (\cref{thm:external_iwm}) on the minimum distance of some majority-supported proposal from an IWM, independent of the weight profile. As we will see in \cref{thm:indiv_weights_positive,thm:indiv_weights_negative}, the severity of Anscombe's Paradox under \emph{internal} weights is closely related to the maximum average topic weight $\tilde{w}_{max}$ (the maximum entry in $\tilde{w}$). 
Formally, we will upper bound the worst-case \emph{IWM distance} $g_{\ell}$ for instances with maximum average topic weight $\tilde{w}_{max} = \ell \in (0, 1)$ and selections of $p_{IWM}$ for the instance:
\begin{align*}
    g_{\ell} := \max_{\substack{\votinginstance = (\prefprof, W),\ \ p_{IWM}\\ s.t. \tilde{w}_{max} = \ell}} \left(\min_{p \text{ weakly majority-supported}} d_H(p, p_{IWM}, \tilde{w})\right)
\end{align*}

We first give a simple upper bound on $g_{\ell}$ for $\ell \in (0,1)$ derived from a partition-based algorithm. Surprisingly, we then show that this seemingly weak upper bound is tight for a large portion of the range $\tilde{w}_{max} \in (0,1)$. Our lower-bound constructions more strongly imply the existence of instances where \emph{all} weakly majority-supported proposals are
far from \emph{all} IWM's. \cref{fig:gwmax} provides a summary of the bounds we give on $g_{\tilde{w}_{max}}$.

\begin{figure}[t]
    \centering
    \includegraphics[width=0.8\columnwidth]{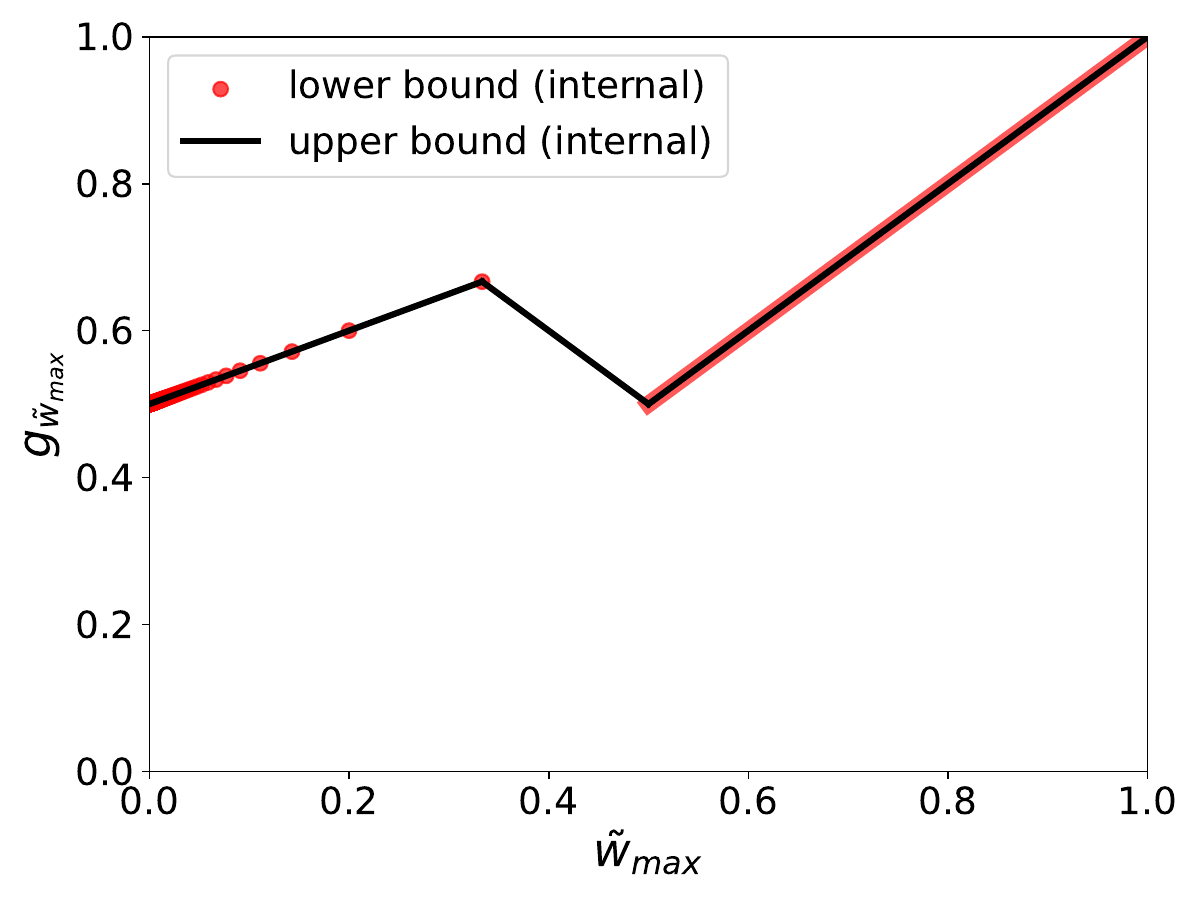}
    \caption{A summary of our bounds on $g_{\tilde{w}_{max}}$.}
    \label{fig:gwmax}
\end{figure}

\textbf{Partition-based upper bounds.} \cref{thm:indiv_weights_positive} guarantees both the existence of reasonable majority-supported proposals and provides an algorithm to efficiently recover them. 
\begin{restatable}{theorem}{internaliwmpos}\label{thm:indiv_weights_positive}
 We have the following upper bounds on $g_{\ell}$:\begin{itemize}
\item If $\ell \in (0, 1/3)$, then $g_{\ell} \leq 1/2 + \ell/2$;
\item If $\ell \in [1/3, 1/2]$, then $g_{\ell} \leq 1- \ell$;
\item If $\ell \in (1/2, 1)$, then $g_{\ell} \leq \ell$.
\end{itemize} 
In each case, we can compute a weakly majority-supported proposal $p$ with $d_H(p, p_{IWM}, \tilde{w})$ at most the given bound in polynomial time.
\end{restatable}

The full proof is deferred to ~\cref{app:positive_internal_iwm}, but the intuition is as follows: for any proposal, either it or its complement will get weak majority support (potentially both), and for any $p_{IWM}$, $d_H(\overline{p}, p_{IWM}, \tilde{w}) = 1 - d_H(p, p_{IWM},\tilde{w})$. Therefore, we construct $p$ that keeps $\max\{d_H(\overline{p}, p_{IWM}, \tilde{w}), d_H(p, p_{IWM},\tilde{w})\}$ small. This is ultimately equivalent to the partition optimization problem with the $t$ entries in the average weight vector as inputs.
Our bounds are constructive and give the pair $(p, \overline{p})$ achieving the bound.

 \textbf{Lower bounds.}
 By definition, $g_{\ell} \leq 1$, so the upper bounds in \cref{thm:indiv_weights_positive} might seem fairly weak. However, in \cref{thm:indiv_weights_negative}, we show that they are actually tight for many values of $\ell$. This implies that when $\tilde{w}_{max}$ is large, $g_{\tilde{w}_{max}}$ can get arbitrarily close to 1. 
 \begin{restatable}{theorem}{internalnegative}\label{thm:indiv_weights_negative}
The following lower bounds for $g_{\ell}$ hold:\begin{itemize}
 \item If $\ell = 1/(2k+1)$ with $k \in \mathbb{Z}_{\geq 0}$, then $g_{\ell} \geq 1/2 + \ell/2$;
 \item If $\ell \in (1/2, 1)$, then $g_{\ell} \geq \ell$.
 \end{itemize}
 \end{restatable}
 We conjecture that the upper bounds in \cref{thm:indiv_weights_positive} are tight for the remaining values of $\ell$, but leave this to future work. The proof of \cref{thm:indiv_weights_negative} is deferred to \cref{app:negative_internal_iwm}, but we provide the construction for $\ell \in (1/2, 1)$ and some intuition here. In the instance below, we choose $x$ large enough such that $\tilde{w}_{max} = \ell$ and the first issue holds a strict majority of the weight for all voters. There are $x$ copies of the first voter, and $x+1$ copies of the second.
 \[
\prefprof =\begin{matrix}(x)\  \times \\ (x+1) \ \times  \end{matrix}\begin{bmatrix}
\tikzmarknode{p311}{\texttt{+1}} & \tikzmarknode{p312}{\texttt{+1}} \\
\tikzmarknode{p321}{\texttt{-1}} & \tikzmarknode{p322}{\texttt{+1}} 
\end{bmatrix} \quad
\mathcal{W} =\begin{matrix}(x) \ \times \\ (x+1) \  \times \end{matrix}\begin{bmatrix}
 \tikzmarknode{p311w}{\frac{x+1}{x} \cdot \ell} & \tikzmarknode{p312w}{1 - \frac{x+1}{x} \cdot \ell}\\
 \tikzmarknode{p321w}{\frac{x}{x+1} \cdot \ell} & \tikzmarknode{p322w}{1 - \frac{x}{x+1} \cdot \ell} 
\end{bmatrix}
\]%
\begin{tikzpicture}[overlay,remember picture, shorten >=-3pt, shorten <= -3pt]
\drawcell{p311}{\cone}
\drawcell{p312}{\cone}
\drawcell{p321}{\czero}
\drawcell{p322}{\cone}
\end{tikzpicture}%

 In this instance, all voters are essentially ``single-issue voters'' on the first topic, but the second type of voters split their weight slightly more evenly between the two topics. $+1$ is the weighted majority opinion on the first topic, but any proposal with $+1$ for that topic will not get majority support because voters of the second type will oppose it. Notably, $\mathbf{1}$ is the unique IWM in our constructions, implying there is no majority-supported proposal close to \textit{any} IWM.

\cref{thm:indiv_weights_negative} quashes any hope of improving on \cref{thm:indiv_weights_positive} and proving a similar result to the external weights setting (where $g_\ell < 1/2$ held for any weights profile). Once voters can have distinct weight vectors, increasing $\tilde{w}_{max}$ can make the distance between all majority-supported proposals and IWM proposals arbitrarily large. We conclude this section by characterizing a group of voting instances in which Anscombe's Paradox will not occur.

\textbf{Condition precluding Anscombe's Paradox.} We find that generalizations of Wagner's Rule of Three-Fourths hold in both the external and internal weights settings:

\begin{restatable}{theorem}{threefourths}\label{thm:three-fourths}
    If $\avgmaj \geq 3/4$ then Anscombe's paradox will not occur. Additionally, if $m_j \geq 3/4$ for all $j \in [t]$ in the external weights setting, then Ostrogorski's paradox will not occur.
\end{restatable}

Our proof (deferred to \cref{app:wagner}) follows Wagner's original proof strategy of counting agreement with an IWM in an instance in two ways: column-wise and row-wise, but is modified to account for weights. We get the second part of our claim by using the fact that, under external weights, Ostrogorski's paradox occurs if and only if there is a subset of issues inducing an instance where Anscombe's paradox occurs.

\section{Conclusion and Future Work}
 
We explored how best to represent the will of voters on multiple, separable issues when optimizing for two potentially conflicting ideals: agreement with issue-wise majority and success in pairwise proposal comparisons. Additionally, we augmented previous multi-issue voting models to account for non-uniform and individualized issue importance. We demonstrated that determining whether an IWM is a Condorcet winner is co-NP hard, but provided an efficiently checkable condition under which Ostrogorski's paradox does not occur. We then examined instances where an IWM loses to the opposing proposal (i.e., Anscombe's paradox occurs) and showed how our two weighting models alter our ability to reconcile the two objectives. While we now have a rich understanding of the interaction of these two majoritarian ideals, one could optimize for different notions of representation in the proposal selection. It would be interesting to study variants of maximizing total voter ``satisfaction'' --- the total weight voters have on topics that they agree with the final proposal on (a weighted version of an objective proposed in \cite{fritsch2022price}). On the technical side, our work leaves open a number of interesting questions and gaps: (i) Our \cref{thm:external_iwm} for external weights is only existential. In contrast, in the unweighted setting, \cite{constantinescu2023computing} also provide a polynomial-time method to derandomize the probabilistic argument. Extending this approach to the weighted setting appears generally more challenging but likely feasible in pseudo-polynomial time with slightly more involved techniques. 
(ii) Paper \cite{constantinescu2023computing} also shows a hardness result for the unweighted case: telling whether a proposal achieves more agreement with an IWM than guaranteed by the probabilistic argument is NP-hard. It would be interesting to get a similar result for every fixed weights vector $w$. (iii) We have only succeeded in proving that our bounds in \cref{thm:indiv_weights_positive} are tight for some portion of the range $\tilde{w}_{max} \in (0,1)$.

\begin{acks}
Carmel Baharav was supported by an ETH Excellence Scholarship. We thank Edith Elkind for the simpler proof of \cite[{Theorem 2}]{complexity_deliberative_coalition} that we include in the appendix. We thank Tamio-Vesa Nakajima for helping typeset \cref{fig:moebius} and the many productive discussions surrounding the linear-time recognition algorithm for single-switch preferences (including noting \cref{lemma:cyclic-shift-and-negate}). Finally, we are grateful to the anonymous reviewers for their insightful comments and suggestions that helped improve the quality of the paper.
\end{acks}

\bibliographystyle{ACM-Reference-Format} 
\bibliography{references}


\begin{thebibliography}{39}


\ifx \showCODEN    \undefined \def \showCODEN     #1{\unskip}     \fi
\ifx \showDOI      \undefined \def \showDOI       #1{#1}\fi
\ifx \showISBNx    \undefined \def \showISBNx     #1{\unskip}     \fi
\ifx \showISBNxiii \undefined \def \showISBNxiii  #1{\unskip}     \fi
\ifx \showISSN     \undefined \def \showISSN      #1{\unskip}     \fi
\ifx \showLCCN     \undefined \def \showLCCN      #1{\unskip}     \fi
\ifx \shownote     \undefined \def \shownote      #1{#1}          \fi
\ifx \showarticletitle \undefined \def \showarticletitle #1{#1}   \fi
\ifx \showURL      \undefined \def \showURL       {\relax}        \fi
\providecommand\bibfield[2]{#2}
\providecommand\bibinfo[2]{#2}
\providecommand\natexlab[1]{#1}
\providecommand\showeprint[2][]{arXiv:#2}

\bibitem[\protect\citeauthoryear{Anscombe}{Anscombe}{1976}]%
        {anscombe1976frustration}
\bibfield{author}{\bibinfo{person}{Gertrude Elizabeth~Margaret Anscombe}.} \bibinfo{year}{1976}\natexlab{}.
\newblock \showarticletitle{On frustration of the majority by fulfilment of the majority's will}.
\newblock \bibinfo{journal}{\emph{Analysis}} \bibinfo{volume}{36}, \bibinfo{number}{4} (\bibinfo{year}{1976}), \bibinfo{pages}{161--168}.
\newblock


\bibitem[\protect\citeauthoryear{Benesch, Heim, Schelker, and Schmid}{Benesch et~al\mbox{.}}{2023}]%
        {benesch2023voting}
\bibfield{author}{\bibinfo{person}{Christine Benesch}, \bibinfo{person}{Rino Heim}, \bibinfo{person}{Mark Schelker}, {and} \bibinfo{person}{Lukas Schmid}.} \bibinfo{year}{2023}\natexlab{}.
\newblock \showarticletitle{Do Voting Advice Applications Change Political Behavior?}
\newblock \bibinfo{journal}{\emph{The Journal of Politics}} \bibinfo{volume}{85}, \bibinfo{number}{2} (\bibinfo{year}{2023}), \bibinfo{pages}{684--700}.
\newblock


\bibitem[\protect\citeauthoryear{Beno{\^\i}t and Kornhauser}{Beno{\^\i}t and Kornhauser}{2010}]%
        {benoit2010only}
\bibfield{author}{\bibinfo{person}{Jean-Pierre Beno{\^\i}t} {and} \bibinfo{person}{Lewis~A Kornhauser}.} \bibinfo{year}{2010}\natexlab{}.
\newblock \showarticletitle{Only a dictatorship is efficient}.
\newblock \bibinfo{journal}{\emph{Games and Economic Behavior}} \bibinfo{volume}{70}, \bibinfo{number}{2} (\bibinfo{year}{2010}), \bibinfo{pages}{261--270}.
\newblock


\bibitem[\protect\citeauthoryear{Bensland, Saleh, and Ronchetti}{Bensland et~al\mbox{.}}{2020}]%
        {mlsmartvote}
\bibfield{author}{\bibinfo{person}{Sebastian Bensland}, \bibinfo{person}{Karim Saleh}, {and} \bibinfo{person}{Pietro Ronchetti}.} \bibinfo{year}{2020}\natexlab{}.
\newblock \bibinfo{title}{A Machine Learning Analysis of the Swiss Political Spectrum and Candidate Recommendation Process}.
\newblock
\newblock


\bibitem[\protect\citeauthoryear{Booth and Lueker}{Booth and Lueker}{1976}]%
        {booth_lueker}
\bibfield{author}{\bibinfo{person}{Kellogg~S. Booth} {and} \bibinfo{person}{George~S. Lueker}.} \bibinfo{year}{1976}\natexlab{}.
\newblock \showarticletitle{Testing for the consecutive ones property, interval graphs, and graph planarity using PQ-tree algorithms}.
\newblock \bibinfo{journal}{\emph{J. Comput. System Sci.}} \bibinfo{volume}{13}, \bibinfo{number}{3} (\bibinfo{year}{1976}), \bibinfo{pages}{335--379}.
\newblock
\showISSN{0022-0000}
\urldef\tempurl%
\url{https://doi.org/10.1016/S0022-0000(76)80045-1}
\showDOI{\tempurl}


\bibitem[\protect\citeauthoryear{Brams and Fishburn}{Brams and Fishburn}{1978}]%
        {brams1978approval}
\bibfield{author}{\bibinfo{person}{Steven~J Brams} {and} \bibinfo{person}{Peter~C Fishburn}.} \bibinfo{year}{1978}\natexlab{}.
\newblock \showarticletitle{Approval voting}.
\newblock \bibinfo{journal}{\emph{American Political Science Review}} \bibinfo{volume}{72}, \bibinfo{number}{3} (\bibinfo{year}{1978}), \bibinfo{pages}{831--847}.
\newblock


\bibitem[\protect\citeauthoryear{Brams and Kilgour}{Brams and Kilgour}{2015}]%
        {brams2015satisfaction}
\bibfield{author}{\bibinfo{person}{Steven~J Brams} {and} \bibinfo{person}{D~Marc Kilgour}.} \bibinfo{year}{2015}\natexlab{}.
\newblock \showarticletitle{Satisfaction approval voting}.
\newblock \bibinfo{journal}{\emph{Mathematical and Computational Modeling: With Applications in Natural and Social Sciences, Engineering, and the Arts}} (\bibinfo{year}{2015}), \bibinfo{pages}{273--298}.
\newblock


\bibitem[\protect\citeauthoryear{Brandt, Conitzer, Endriss, Lang, and Procaccia}{Brandt et~al\mbox{.}}{2016}]%
        {brandt2016handbook}
\bibfield{author}{\bibinfo{person}{Felix Brandt}, \bibinfo{person}{Vincent Conitzer}, \bibinfo{person}{Ulle Endriss}, \bibinfo{person}{J{\'e}r{\^o}me Lang}, {and} \bibinfo{person}{Ariel~D Procaccia}.} \bibinfo{year}{2016}\natexlab{}.
\newblock \bibinfo{booktitle}{\emph{Handbook of computational social choice}}.
\newblock \bibinfo{publisher}{Cambridge University Press}.
\newblock


\bibitem[\protect\citeauthoryear{Bredereck, Chen, and Woeginger}{Bredereck et~al\mbox{.}}{2013}]%
        {bredereck_sc_forbidden_minors}
\bibfield{author}{\bibinfo{person}{Robert Bredereck}, \bibinfo{person}{Jiehua Chen}, {and} \bibinfo{person}{Gerhard~J. Woeginger}.} \bibinfo{year}{2013}\natexlab{}.
\newblock \showarticletitle{A characterization of the single-crossing domain}.
\newblock \bibinfo{journal}{\emph{Social Choice and Welfare}} \bibinfo{volume}{41}, \bibinfo{number}{4} (\bibinfo{year}{2013}), \bibinfo{pages}{989--998}.
\newblock
\showISSN{01761714, 1432217X}
\urldef\tempurl%
\url{http://www.jstor.org/stable/42001441}
\showURL{%
\tempurl}


\bibitem[\protect\citeauthoryear{Brighouse and Fleurbaey}{Brighouse and Fleurbaey}{2010}]%
        {brighouse2008democracy}
\bibfield{author}{\bibinfo{person}{Harry Brighouse} {and} \bibinfo{person}{Marc Fleurbaey}.} \bibinfo{year}{2010}\natexlab{}.
\newblock \showarticletitle{Democracy and Proportionality}.
\newblock \bibinfo{journal}{\emph{The Journal of Political Philosophy}} \bibinfo{volume}{18}, \bibinfo{number}{2} (\bibinfo{year}{2010}), \bibinfo{pages}{137--155}.
\newblock


\bibitem[\protect\citeauthoryear{Casella}{Casella}{2005}]%
        {casella2005storable}
\bibfield{author}{\bibinfo{person}{Alessandra Casella}.} \bibinfo{year}{2005}\natexlab{}.
\newblock \showarticletitle{Storable votes}.
\newblock \bibinfo{journal}{\emph{Games and Economic Behavior}} \bibinfo{volume}{51}, \bibinfo{number}{2} (\bibinfo{year}{2005}), \bibinfo{pages}{391--419}.
\newblock


\bibitem[\protect\citeauthoryear{Chan and P\u{a}tra\c{s}cu}{Chan and P\u{a}tra\c{s}cu}{2010}]%
        {patrascu_inversions}
\bibfield{author}{\bibinfo{person}{Timothy~M. Chan} {and} \bibinfo{person}{Mihai P\u{a}tra\c{s}cu}.} \bibinfo{year}{2010}\natexlab{}.
\newblock \showarticletitle{Counting inversions, offline orthogonal range counting, and related problems}. In \bibinfo{booktitle}{\emph{Proceedings of the Twenty-First Annual ACM-SIAM Symposium on Discrete Algorithms}} (Austin, Texas) \emph{(\bibinfo{series}{SODA '10})}. \bibinfo{publisher}{Society for Industrial and Applied Mathematics}, \bibinfo{address}{USA}, \bibinfo{pages}{161--173}.
\newblock
\showISBNx{9780898716986}


\bibitem[\protect\citeauthoryear{Constantinescu and Wattenhofer}{Constantinescu and Wattenhofer}{2023}]%
        {constantinescu2023computing}
\bibfield{author}{\bibinfo{person}{Andrei Constantinescu} {and} \bibinfo{person}{Roger Wattenhofer}.} \bibinfo{year}{2023}\natexlab{}.
\newblock \showarticletitle{Computing the Best Policy That Survives a Vote}.
\newblock \bibinfo{journal}{\emph{Proceedings of the 22nd International Conference on Autonomous Agents and Multiagent Systems (AAMAS `23)}} (\bibinfo{year}{2023}).
\newblock


\bibitem[\protect\citeauthoryear{Elkind, Ghosh, and Goldberg}{Elkind et~al\mbox{.}}{2022a}]%
        {complexity_deliberative_coalition}
\bibfield{author}{\bibinfo{person}{Edith Elkind}, \bibinfo{person}{Abheek Ghosh}, {and} \bibinfo{person}{Paul~W. Goldberg}.} \bibinfo{year}{2022}\natexlab{a}.
\newblock \showarticletitle{Complexity of Deliberative Coalition Formation}. In \bibinfo{booktitle}{\emph{Proceedings of the AAAI Conference on Artificial Intelligence}}, Vol.~\bibinfo{volume}{36}. \bibinfo{pages}{4975--4982}.
\newblock


\bibitem[\protect\citeauthoryear{Elkind and Lackner}{Elkind and Lackner}{2015}]%
        {elkind_lackner_dichotomous}
\bibfield{author}{\bibinfo{person}{Edith Elkind} {and} \bibinfo{person}{Martin Lackner}.} \bibinfo{year}{2015}\natexlab{}.
\newblock \showarticletitle{Structure in dichotomous preferences}. In \bibinfo{booktitle}{\emph{Proceedings of the 24th International Conference on Artificial Intelligence}} (Buenos Aires, Argentina) \emph{(\bibinfo{series}{IJCAI'15})}. \bibinfo{publisher}{AAAI Press}, \bibinfo{pages}{2019--2025}.
\newblock
\showISBNx{9781577357384}


\bibitem[\protect\citeauthoryear{Elkind, Lackner, and Peters}{Elkind et~al\mbox{.}}{2022b}]%
        {survey_restricted}
\bibfield{author}{\bibinfo{person}{Edith Elkind}, \bibinfo{person}{Martin Lackner}, {and} \bibinfo{person}{Dominik Peters}.} \bibinfo{year}{2022}\natexlab{b}.
\newblock \bibinfo{title}{Preference Restrictions in Computational Social Choice: A Survey}.
\newblock
\newblock


\bibitem[\protect\citeauthoryear{Fink, Pfretzschner, and Rutter}{Fink et~al\mbox{.}}{2023}]%
        {experimental_pq_trees}
\bibfield{author}{\bibinfo{person}{Simon~D. Fink}, \bibinfo{person}{Matthias Pfretzschner}, {and} \bibinfo{person}{Ignaz Rutter}.} \bibinfo{year}{2023}\natexlab{}.
\newblock \showarticletitle{Experimental Comparison of PC-Trees and PQ-Trees}.
\newblock \bibinfo{journal}{\emph{ACM J. Exp. Algorithmics}}  \bibinfo{volume}{28}, Article \bibinfo{articleno}{1.10} (\bibinfo{date}{oct} \bibinfo{year}{2023}).
\newblock
\showISSN{1084-6654}
\urldef\tempurl%
\url{https://doi.org/10.1145/3611653}
\showDOI{\tempurl}


\bibitem[\protect\citeauthoryear{Fishburn}{Fishburn}{1981}]%
        {fishburn1981analysis}
\bibfield{author}{\bibinfo{person}{Peter~C Fishburn}.} \bibinfo{year}{1981}\natexlab{}.
\newblock \showarticletitle{An analysis of simple voting systems for electing committees}.
\newblock \bibinfo{journal}{\emph{SIAM J. Appl. Math.}} \bibinfo{volume}{41}, \bibinfo{number}{3} (\bibinfo{year}{1981}), \bibinfo{pages}{499--502}.
\newblock


\bibitem[\protect\citeauthoryear{Flanigan, Procaccia, and Wang}{Flanigan et~al\mbox{.}}{2023}]%
        {flaniganaccounting}
\bibfield{author}{\bibinfo{person}{Bailey Flanigan}, \bibinfo{person}{Ariel~D Procaccia}, {and} \bibinfo{person}{Sven Wang}.} \bibinfo{year}{2023}\natexlab{}.
\newblock \bibinfo{title}{Accounting for Stakes in Democratic Decisions}.  (\bibinfo{year}{2023}).
\newblock
\newblock
\shownote{FORC'23 (non-archival).}


\bibitem[\protect\citeauthoryear{Fritsch and Wattenhofer}{Fritsch and Wattenhofer}{2022}]%
        {fritsch2022price}
\bibfield{author}{\bibinfo{person}{Robin Fritsch} {and} \bibinfo{person}{Roger Wattenhofer}.} \bibinfo{year}{2022}\natexlab{}.
\newblock \showarticletitle{The Price of Majority Support}.
\newblock \bibinfo{journal}{\emph{Proceedings of the 21st International Conference on Autonomous Agents and Multiagent Systems (AAMAS `22)}} (\bibinfo{year}{2022}).
\newblock


\bibitem[\protect\citeauthoryear{Garey and Johnson}{Garey and Johnson}{1979}]%
        {garey1979computers}
\bibfield{author}{\bibinfo{person}{Michael~R. Garey} {and} \bibinfo{person}{David~S. Johnson}.} \bibinfo{year}{1979}\natexlab{}.
\newblock \bibinfo{booktitle}{\emph{Computers and Intractability: A Guide to the Theory of NP-Completeness}}.
\newblock


\bibitem[\protect\citeauthoryear{Grandi}{Grandi}{2014}]%
        {grandi2014common}
\bibfield{author}{\bibinfo{person}{Umberto Grandi}.} \bibinfo{year}{2014}\natexlab{}.
\newblock \showarticletitle{The common structure of paradoxes in aggregation theory}.
\newblock \bibinfo{journal}{\emph{arXiv preprint arXiv:1406.2855}} (\bibinfo{year}{2014}).
\newblock


\bibitem[\protect\citeauthoryear{Kornhauser and Sager}{Kornhauser and Sager}{1993}]%
        {kornhauser1993one}
\bibfield{author}{\bibinfo{person}{Lewis~A Kornhauser} {and} \bibinfo{person}{Lawrence~G Sager}.} \bibinfo{year}{1993}\natexlab{}.
\newblock \showarticletitle{The one and the many: Adjudication in collegial courts}.
\newblock \bibinfo{journal}{\emph{Calif. L. Rev.}}  \bibinfo{volume}{81} (\bibinfo{year}{1993}), \bibinfo{pages}{1}.
\newblock


\bibitem[\protect\citeauthoryear{Lackner and Skowron}{Lackner and Skowron}{2022}]%
        {lackner2022approval}
\bibfield{author}{\bibinfo{person}{Martin Lackner} {and} \bibinfo{person}{Piotr Skowron}.} \bibinfo{year}{2022}\natexlab{}.
\newblock \showarticletitle{Approval-based committee voting}.
\newblock In \bibinfo{booktitle}{\emph{Multi-Winner Voting with Approval Preferences}}. \bibinfo{publisher}{Springer}, \bibinfo{pages}{1--7}.
\newblock


\bibitem[\protect\citeauthoryear{Laffond and Lain{\'e}}{Laffond and Lain{\'e}}{2006}]%
        {laffond2006single}
\bibfield{author}{\bibinfo{person}{Gilbert Laffond} {and} \bibinfo{person}{Jean Lain{\'e}}.} \bibinfo{year}{2006}\natexlab{}.
\newblock \showarticletitle{Single-switch preferences and the Ostrogorski paradox}.
\newblock \bibinfo{journal}{\emph{Mathematical Social Sciences}} \bibinfo{volume}{52}, \bibinfo{number}{1} (\bibinfo{year}{2006}), \bibinfo{pages}{49--66}.
\newblock


\bibitem[\protect\citeauthoryear{Laffond and Lain{\'e}}{Laffond and Lain{\'e}}{2013}]%
        {laffond2013unanimity}
\bibfield{author}{\bibinfo{person}{Gilbert Laffond} {and} \bibinfo{person}{Jean Lain{\'e}}.} \bibinfo{year}{2013}\natexlab{}.
\newblock \showarticletitle{Unanimity and the Anscombe’s paradox}.
\newblock \bibinfo{journal}{\emph{Top}}  \bibinfo{volume}{21} (\bibinfo{year}{2013}), \bibinfo{pages}{590--611}.
\newblock


\bibitem[\protect\citeauthoryear{Lalley and Weyl}{Lalley and Weyl}{2018}]%
        {lalley2018quadratic}
\bibfield{author}{\bibinfo{person}{Steven~P Lalley} {and} \bibinfo{person}{E~Glen Weyl}.} \bibinfo{year}{2018}\natexlab{}.
\newblock \showarticletitle{Quadratic voting: How mechanism design can radicalize democracy}. In \bibinfo{booktitle}{\emph{AEA Papers and Proceedings}}, Vol.~\bibinfo{volume}{108}. American Economic Association 2014 Broadway, Suite 305, Nashville, TN 37203, \bibinfo{pages}{33--37}.
\newblock


\bibitem[\protect\citeauthoryear{Lang}{Lang}{2004}]%
        {lang2004logical}
\bibfield{author}{\bibinfo{person}{J{\'e}r{\^o}me Lang}.} \bibinfo{year}{2004}\natexlab{}.
\newblock \showarticletitle{Logical preference representation and combinatorial vote}.
\newblock \bibinfo{journal}{\emph{Annals of Mathematics and Artificial Intelligence}}  \bibinfo{volume}{42} (\bibinfo{year}{2004}), \bibinfo{pages}{37--71}.
\newblock


\bibitem[\protect\citeauthoryear{List}{List}{2003}]%
        {list2003possibility}
\bibfield{author}{\bibinfo{person}{Christian List}.} \bibinfo{year}{2003}\natexlab{}.
\newblock \showarticletitle{A possibility theorem on aggregation over multiple interconnected propositions}.
\newblock \bibinfo{journal}{\emph{Mathematical Social Sciences}} \bibinfo{volume}{45}, \bibinfo{number}{1} (\bibinfo{year}{2003}), \bibinfo{pages}{1--13}.
\newblock


\bibitem[\protect\citeauthoryear{List}{List}{2012}]%
        {list2012theory}
\bibfield{author}{\bibinfo{person}{Christian List}.} \bibinfo{year}{2012}\natexlab{}.
\newblock \showarticletitle{The theory of judgment aggregation: an introductory review}.
\newblock \bibinfo{journal}{\emph{Synthese}} \bibinfo{volume}{187}, \bibinfo{number}{1} (\bibinfo{year}{2012}), \bibinfo{pages}{179--207}.
\newblock


\bibitem[\protect\citeauthoryear{Matsui and Matsui}{Matsui and Matsui}{2000}]%
        {matsui2000survey}
\bibfield{author}{\bibinfo{person}{Tomomi Matsui} {and} \bibinfo{person}{Yasuko Matsui}.} \bibinfo{year}{2000}\natexlab{}.
\newblock \showarticletitle{A survey of algorithms for calculating power indices of weighted majority games}.
\newblock \bibinfo{journal}{\emph{Journal of the Operations Research Society of Japan}} \bibinfo{volume}{43}, \bibinfo{number}{1} (\bibinfo{year}{2000}), \bibinfo{pages}{71--86}.
\newblock


\bibitem[\protect\citeauthoryear{Nurmi}{Nurmi}{1999}]%
        {nurmi_compound_paradoxes}
\bibfield{author}{\bibinfo{person}{Hannu Nurmi}.} \bibinfo{year}{1999}\natexlab{}.
\newblock \showarticletitle{Compound Majority Paradoxes}.
\newblock In \bibinfo{booktitle}{\emph{Voting Paradoxes and How to Deal with Them}}. \bibinfo{publisher}{Springer Berlin Heidelberg}, \bibinfo{pages}{70--86}.
\newblock
\showISBNx{978-3-662-03782-9}


\bibitem[\protect\citeauthoryear{{\"O}zkal-Sanver and Sanver}{{\"O}zkal-Sanver and Sanver}{2006}]%
        {ozkal2006ensuring}
\bibfield{author}{\bibinfo{person}{Ipek {\"O}zkal-Sanver} {and} \bibinfo{person}{M~Remzi Sanver}.} \bibinfo{year}{2006}\natexlab{}.
\newblock \showarticletitle{Ensuring pareto optimality by referendum voting}.
\newblock \bibinfo{journal}{\emph{Social Choice and Welfare}}  \bibinfo{volume}{27} (\bibinfo{year}{2006}), \bibinfo{pages}{211--219}.
\newblock


\bibitem[\protect\citeauthoryear{Rae and Daudt}{Rae and Daudt}{1976}]%
        {rae1976ostrogorski}
\bibfield{author}{\bibinfo{person}{Douglas~W Rae} {and} \bibinfo{person}{Hans Daudt}.} \bibinfo{year}{1976}\natexlab{}.
\newblock \showarticletitle{The Ostrogorski paradox: a peculiarity of compound majority decision}.
\newblock \bibinfo{journal}{\emph{European Journal of Political Research}} \bibinfo{volume}{4}, \bibinfo{number}{4} (\bibinfo{year}{1976}), \bibinfo{pages}{391--398}.
\newblock


\bibitem[\protect\citeauthoryear{Research}{Research}{2023}]%
        {pewpoll}
\bibfield{author}{\bibinfo{person}{Pew Research}.} \bibinfo{year}{2023}\natexlab{}.
\newblock \bibinfo{title}{Inflation, Health Costs, Partisan Cooperation Among the Nation’s Top Problems}.
\newblock
\newblock
\urldef\tempurl%
\url{https://www.pewresearch.org/politics/2023/06/21/inflation-health-costs-partisan-cooperation-among-the-nations-top-problems/}
\showURL{%
\tempurl}
\newblock
\shownote{Last accessed 19 September 2023.}


\bibitem[\protect\citeauthoryear{Terzopoulou, Karpov, and Obraztsova}{Terzopoulou et~al\mbox{.}}{2021}]%
        {Terzopoulou_Karpov_Obraztsova_2021}
\bibfield{author}{\bibinfo{person}{Zoi Terzopoulou}, \bibinfo{person}{Alexander Karpov}, {and} \bibinfo{person}{Svetlana Obraztsova}.} \bibinfo{year}{2021}\natexlab{}.
\newblock \showarticletitle{Restricted Domains of Dichotomous Preferences with Possibly Incomplete Information}.
\newblock \bibinfo{journal}{\emph{Proceedings of the AAAI Conference on Artificial Intelligence}} \bibinfo{volume}{35}, \bibinfo{number}{6} (\bibinfo{date}{May} \bibinfo{year}{2021}), \bibinfo{pages}{5726--5733}.
\newblock
\urldef\tempurl%
\url{https://doi.org/10.1609/aaai.v35i6.16718}
\showDOI{\tempurl}


\bibitem[\protect\citeauthoryear{Uckelman}{Uckelman}{2009}]%
        {uckelman2009more}
\bibfield{author}{\bibinfo{person}{Joel Uckelman}.} \bibinfo{year}{2009}\natexlab{}.
\newblock \bibinfo{booktitle}{\emph{More than the sum of its parts: compact preference representation over combinatorial domains}}.
\newblock \bibinfo{publisher}{University of Amsterdam}.
\newblock


\bibitem[\protect\citeauthoryear{Wagner}{Wagner}{1983}]%
        {wagner1983anscombe}
\bibfield{author}{\bibinfo{person}{Carl Wagner}.} \bibinfo{year}{1983}\natexlab{}.
\newblock \showarticletitle{Anscombe's paradox and the rule of three-fourths}.
\newblock \bibinfo{journal}{\emph{Theory and Decision}} \bibinfo{volume}{15}, \bibinfo{number}{3} (\bibinfo{year}{1983}), \bibinfo{pages}{303}.
\newblock


\bibitem[\protect\citeauthoryear{Wagner}{Wagner}{1984}]%
        {wagner1984avoiding}
\bibfield{author}{\bibinfo{person}{Carl Wagner}.} \bibinfo{year}{1984}\natexlab{}.
\newblock \showarticletitle{Avoiding Anscombe's paradox}.
\newblock \bibinfo{journal}{\emph{Theory and Decision}} \bibinfo{volume}{16}, \bibinfo{number}{3} (\bibinfo{year}{1984}), \bibinfo{pages}{233}.
\newblock


\end{thebibliography}


\newpage

\appendix

\section{Complexity of Determining a Condorcet Winner}\label{app:condorcet}

In this section, we show that \fauxsc{Unanim} is NP-hard. This was previously known \cite[{Theorem 2}]{complexity_deliberative_coalition}, but the proof there is a more complicated reduction from Independent Set. In contrast, our proof proceeds by recasting the problem in terms of selecting a subset of column vectors whose sum is negative in all coordinates. This alternate view enables a more natural reduction from Exact Cover By 3-Sets. Afterward, we put all the pieces together to give a self-contained formal proof of \cref{th:checking-condorcet-co-np-hard}. We conclude the section by providing a similar ``choosing a subset of vectors'' formulation for \fauxsc{Major}, which we believe could be interesting more broadly.

To prove that \fauxsc{Unanim} is NP-hard, we note that \fauxsc{Unanim} admits an elegant reformulation: write $\prefprof = (c_1, \ldots, c_t)$ in terms of its issue/column vectors, then selecting a proposal $p \in \Bset^t$ amounts to choosing the subset of issues where $p$ differs from $\mathbf{1}$. A voter $i \in [n]$ prefers $p \succ_i \mathbf{1}$ iff the sum of the selected vectors is strictly negative in coordinate $i$, so \fauxsc{Unanim} asks to select vectors such that the sum is negative in all $n$ coordinates.
Hence, \fauxsc{Unanim} is equivalent to the following problem, which we find of independent interest:

\begin{tcolorbox}[boxsep=0mm]
\textbf{Problem ``\fauxsc{Negative-Sum-Subset}'' (\fauxsc{NSS})}\\
\textbf{Input}: Collection $C$ of $t$ dimension-$n$ vectors over $\pm 1$. \\
\textbf{Output}: Does there exist 
$C' \subseteq C$ s.t~$\textbf{v} := \sum_{c \in C'} c$ is negative in all coordinates: $\textbf{v}_i < 0$ for all $i \in [n]$? 
\end{tcolorbox}

\begin{lemma}\label{lemma:unanim-is-np-hard} \fauxsc{UNANIM} = \fauxsc{NSS} is NP-hard.
\end{lemma}
\begin{proof}
    We reduce from the problem Exact Cover By 3-Sets (\fauxsc{X3C}), which is well-known to be NP-hard \cite{garey1979computers}. An instance of \fauxsc{X3C} is given by a ground set $S$ of size $3s$ and a set $X$ of size-3 subsets of S (also called \emph{triples}); it is a yes-instance if and only if there exists a subset $X' \subseteq X$ that forms an exact cover of $S$ (i.e., the triples in $X'$ cover each element of $S$ exactly once); note that $|X'| = s$ must hold if so.

    Consider an instance of \fauxsc{X3C}, we want to construct an instance of \fauxsc{NSS} such that the former is a yes-instance if and only if the latter one is.
    
    We can assume without loss of generality that 
    there is a triple $\{a, b, c\} \in X$ such that no other element of $X$ contains any of $a$, $b$, or $c$. Indeed, if not, we can add three new elements to $S$ and a set comprising them to $X$. The modified instance is a yes-instance if and only if the original instance is.

    Now, to create an instance of \fauxsc{NSS}, we construct $|X|+s-1$ dimension-$(3s+1)$ vectors over $\pm1.$ The first $3s$ coordinates correspond to the elements of the ground set. For each set $x \in X$, we construct a vector with $-1$ in coordinates that correspond to elements of $x$ as well as in the last coordinate (in other coordinates, we have $+1$). We refer to these vectors as \emph{set vectors}. In addition, we construct $s-1$ vectors with $-1$ in the first $3s$ coordinates and $+1$ in the last coordinate; we refer to these vectors as \emph{dummy vectors}. 

    We claim that this is a yes-instance of \fauxsc{NSS} if and only if the original instance of \fauxsc{X3C} was a yes-instance.

    $(\Leftarrow)$ Suppose we started with a yes-instance of \fauxsc{X3C}, so let $X' \subseteq X$ be an exact cover. Recall that this implies $|X'| = s$. Take the $s$ set vectors corresponding to $X'$ as well as all $s-1$ dummy vectors. Summing up these vectors, in the last coordinate, we get $s \cdot (-1) + (s-1) \cdot (+1) = -1$.
    In every other coordinate, we have $s-1$ entries $-1$ from the dummy vectors, one entry $-1$ from the vector corresponding to the triple that covers the respective element, and $s-1$ entries $+1$ from the other selected set vectors, making for a total sum of $(s-1) \cdot (-1) + 1 \cdot (-1) + (s - 1) \cdot (+1) = -1$. Hence, the resulting instance is a yes-instance of \fauxsc{NSS}.

    $(\Rightarrow)$ Conversely, suppose the resulting instance is a yes-instance of \fauxsc{NSS}. Hence, pick some of the vectors so that the sum in each coordinate is negative. We will show that, among picked vectors, the set vectors correspond to an exact cover. Suppose we picked $t$ set vectors. Then, we picked at most $t - 1$ dummy vectors (as otherwise, the sum in the last coordinate would be non-negative). Note that this implies $t > 0$.

    Assume for a contradiction that the set vectors among selected vectors do not correspond to a cover. Then, there is some coordinate among the first $3s$ in which all selected set vectors have $+1$. However, in that case, the sum of the vectors cannot be negative in that coordinate, as we have $t$ entries $+1$ from set vectors and at most $t-1$ entries $-1$ from dummy vectors, a contradiction.
    
    On the other hand, assume for a contradiction that the selected set vectors do correspond to a cover, but it is not an exact cover, so some element is covered more than once. In particular, each element is covered at least once, and at least one element element is covered more than once, so the sum of the sizes of the triples that the selected set vectors correspond to has to be at least $3s + 1,$ meaning that we had to pick at least $s+1 = \lceil \frac{3s + 1}{3} \rceil$ set vectors.
    Now, recall that there is some element $a$ (without loss of generality, the element corresponding to the first coordinate) that is only contained in one triple in X.\footnote{In fact, there are three such elements, but we do not need this fact.} Then, in our chosen collection of vectors, we have at most one set vector with $-1$ in the first coordinate and, hence, at least $s$ set vectors with $+1$ in this coordinate. Hence, considering only set vectors, the sum in this coordinate is at least $1 \cdot (-1) + s \cdot (+1) = s - 1$. Hence, even if all $s-1$ dummy vectors were picked, the sum in this coordinate is at least $s - 1 + (s - 1) \cdot (-1) = 0$, and hence non-negative, a contradiction. 
    
    Hence, the set vectors among picked vectors correspond to an exact cover, so the original instance of \fauxsc{X3C} is a yes-instance, completing the proof. 
\end{proof}

We now give a self-contained formal proof of \cref{th:checking-condorcet-co-np-hard}, restated below for convenience.

\condorcet*

\begin{proof}%
In the unweighted setting with odd $n$, checking for the existence of a Condorcet winner is equivalent to checking whether the (unique) issue-wise majority proposal is a Condorcet winner (by \cref{lemma:condorcet-implies-iwm}). It is enough to show hardness for the case where the issue-wise majority proposal is $\mathbf{1}$.\footnote{This is, in fact, equivalent to the general case by negating/flipping issues where $-1$ is the majority opinion, but we do not need this distinction to prove hardness.} Hence, we want to show that the following problem is co-NP-hard: ``Given a voting instance $\votinginstance = \prefprof$ in the unweighted setting with odd $n$ where $\mathbf{1}$ is the issue-wise majority proposal, determine whether $\mathbf{1}$ is a Condorcet winner.'' This is equivalent to showing that the `negated' problem is NP-hard: ``Given a voting instance $\votinginstance = \prefprof$ in the unweighted setting with odd $n$ where $\mathbf{1}$ is the issue-wise majority proposal, determine whether $\mathbf{1}$ is \emph{not} a Condorcet winner.'' Not being a Condorcet winner is equivalent to there existing a proposal $p \in \Bset^t$ such that $p \succ_\votinginstance \mathbf{1}$. Hence, the negated problem is precisely \fauxsc{Major}, which is NP-hard by \cref{lemma:major-is-np-hard}, completing the proof.
\end{proof}

We end by pointing out that it is also possible to formulate \fauxsc{Major} as a vector problem. We found it clearest not to do so in our proofs, but the problem resulting in doing so is elegant and might see other applications. Namely, the following problem is equivalent to \fauxsc{Major}, and hence NP-hard:

\begin{tcolorbox}[boxsep=0mm]
\textbf{Problem \\ ``\fauxsc{More-Negative-Than-Positive-Sum-Subset}''}\\ 
\textbf{Input}: Collection $C$ of $t$ dimension-$n$ vectors over $\pm 1$, s.t.~$n$ is odd and each vector's entries sum up to a positive amount. \\
\textbf{Output}: Does there exist $C' \subseteq C$ s.t.~$\textbf{v} := \sum_{c \in C'} c$ is negative in more coordinates than it is positive: $\sum_{i \in [n]}\sgn{\textbf{v}_i} < 0$?
\end{tcolorbox}

\section{An Ostrogorski-free Domain}
\subsection{For External Weights Single-Switch Prevents Ostrogorski's Paradox}\label{app:external-weights-single-switch}

In this section, we prove that, under external weights, the single-switch condition guarantees that all IWM proposals are Condorcet winners, i.e., Ostrogorski's paradox does not occur. To this end, we first show that every issue-wise majority proposal does not lose against its opposite, i.e., Anscombe's paradox does not occur. This was already known for the unweighted model \cite{laffond2006single}. We, however, found the proof in \cite{laffond2006single} relatively tricky to parse: in part because of certain missing/unclear details and in part because of the relatively involved case distinctions in the second part of the proof. We adapt this proof to external weights and also simplify and streamline it, removing the need for case distinctions by noting symmetries and wisely manipulating $\sgnText$ functions. This yields a shorter and clearer argument that highlights better where the assumption about our proposal being an IWM comes into play, which was not immediate from the original presentation.

\begin{restatable}{lemma}{anscombesingleswitch}\label{lemma:anscombe-does-not-occur-for-single-switch} Consider an external-weights single-switch instance $\votinginstance = (\prefprof, w)$. Then, for any issue-wise majority proposal $p$, we have $p \succcurlyeq_\votinginstance \overline{p}$.
\end{restatable}

\begin{proof} Without loss of generality, assume that ones form a prefix or a suffix on every row of $\prefprof$.\footnote{Note that assuming this and, at the same time, that $\mathbf{1}$ is an issue-wise majority proposal would lose generality.} For each $1 \leq i \leq t$ define $e_i \in \Bset^t$ such that $e_{i, j} = +1$ if $j \leq i$, and $e_{i, j} = -1$ otherwise. It follows that every row in $\prefprof$ belongs to the set $\cup_{i = 1}^{t} \{e_i, \overline{e_i}\}$. Given a vector $\mathbf{u} \in \Bset^t$ write $\#_\prefprof(\mathbf{u}) := |\{i \in [n] \mid v_i = \mathbf{u} \}|$ for the number of voters in $\prefprof$ whose vote is $\mathbf{u}$, in which case we say that those voters are of \emph{type} $\mathbf{u}$. For $1 \leq i \leq t$, define $x_i := \#_\prefprof(e_i) - \#_\prefprof(\overline{e_i})$ to be the number of voters in $\prefprof$ with vote $e_i$ minus the number of voters in $\prefprof$ with vote $\overline{e_i}$. We want to show that $p \succcurlyeq_\votinginstance \overline{p}$, which amounts to $\sum_{i = 1}^n \sgn{\innerprod{v_i}{p - \overline{p}}_w} \geq 0$, which is equivalent to $\sum_{i = 1}^n \sgn{\innerprod{v_i}{p}_w} \geq 0$. Since voters can only be of the $2t$ types, this is equivalent to:
\[ 
\sum_{i = 1}^t \left(\#_\prefprof(e_i) \cdot \sgn{\innerprod{e_i}{p}_w} + \#_\prefprof(\overline{e_i}) \cdot \sgn{\innerprod{\overline{e_i}}{p}_w}\right) \geq 0
\]
Since $\innerprod{\overline{e_i}}{p}_w = - \innerprod{e_i}{p}_w$, this is the same as:
\begin{equation}\label{eq:1}
\sum_{i = 1}^t x_i \cdot \sgn{\innerprod{e_i}{p}_w} \geq 0 
\end{equation}
To prove this, we are given the fact that $p$ is an issue-wise majority proposal. Recall that for $j \in [t]$ we defined $b_j := \sum_{i = 1}^n v_{i, j} = \sum_{i = 1}^{t} x_i \cdot e_{i, j}$, in which case what we know amounts to $b_j \cdot p_j \geq 0$ for all $j \in [t]$.

\begin{figure}[t!]
\centering
\begin{tabular}{c | ccc@{\hskip 0.2in}c}
      & $e_1'$ & $e_2'$ & $e_3'$ & $e_4'$  \\
      \hline
       \tikzmarknode{cs_1_1}{$e_0$} & 
       \tikzmarknode{cs_1_2}{\texttt{-1}} & \tikzmarknode{cs_1_3}{\texttt{-1}} & \tikzmarknode{cs_1_4}{\texttt{-1}} & \tikzmarknode{cs_1_5}{\texttt{  }} \\[.06in]
       \tikzmarknode{cs_2_1}{$e_1$} & 
       \tikzmarknode{cs_2_2}{\texttt{+1}} & \tikzmarknode{cs_2_3}{\texttt{-1}} & \tikzmarknode{cs_2_4}{\texttt{-1}} & \tikzmarknode{cs_2_5}{\texttt{-1}} \\
       \tikzmarknode{cs_3_1}{$e_2$} & 
       \tikzmarknode{cs_3_2}{\texttt{+1}} & \tikzmarknode{cs_3_3}{\texttt{+1}} & \tikzmarknode{cs_3_4}{\texttt{-1}} & \tikzmarknode{cs_3_5}{\texttt{-1}} \\
       \tikzmarknode{cs_4_1}{$e_3$} & 
       \tikzmarknode{cs_4_2}{\texttt{+1}} & \tikzmarknode{cs_4_3}{\texttt{+1}} & \tikzmarknode{cs_4_4}{\texttt{+1}} & \tikzmarknode{cs_4_5}{\texttt{-1}} \\
\end{tabular}
\caption{Matrix in the proof of \cref{lemma:anscombe-does-not-occur-for-single-switch} for $t = 3$. The helper $e_0$ and $e_4'$ are depicted above/to the right.}
\label{fig:proof-matrix}
\begin{tikzpicture}[overlay,remember picture, shorten >=-3pt, shorten <= -3pt]
\drawcell{cs_1_2}{\czero}
\drawcell{cs_1_3}{\czero}
\drawcell{cs_1_4}{\czero}
\drawcell{cs_2_2}{\cone}
\drawcell{cs_2_3}{\czero}
\drawcell{cs_2_4}{\czero}
\drawcell{cs_2_5}{\czero}
\drawcell{cs_3_2}{\cone}
\drawcell{cs_3_3}{\cone}
\drawcell{cs_3_4}{\czero}
\drawcell{cs_3_5}{\czero}
\drawcell{cs_4_2}{\cone}
\drawcell{cs_4_3}{\cone}
\drawcell{cs_4_4}{\cone}
\drawcell{cs_4_5}{\czero}
\end{tikzpicture}
\end{figure}

To begin the proof, note that the type votes $e_1, \ldots, e_t$ form a $t \times t$ matrix with one row for each vote. Let $e_1', \ldots, e_t'$ be the columns of this matrix. 
For uniformity in the reasoning that follows, it will be helpful to define $e_{t + 1}' := \overline{e_1'} = -\mathbf{1}$ and $e_0 := \overline{e_t} = -\mathbf{1}$. See \cref{fig:proof-matrix} for an illustration. With these added conventions, for any $j \in [t]$ we have that $e_j'$ and $e_{j + 1}'$ differ only in coordinate $j$, which is $+1$ in $e_j'$ and $-1$ in $e_{j + 1}'$, and, moreover,
$e_j$ and $e_{j - 1}$ differ only in coordinate $j$, which is $+1$ in $e_j$ and $-1$ in $e_{j - 1}$. Let us also extend the definition of $b$ to make $b_{t + 1}$ well-defined. 

As a result, for all $i \in [t]$ we have that $b_i - b_{i + 1} = 2 \cdot x_i$ and $\innerprod{e_i - e_{i - 1}}{p}_w = 2 \cdot w_i \cdot p_i$. Moreover, by definition, $b_{t + 1} = -b_{1}$ and $\innerprod{e_t}{p}_w = -\innerprod{e_0}{p}_w$, which together imply that $b_{t + 1} \cdot \sgn{\innerprod{e_t}{p}_w} = b_1 \cdot \sgn{\innerprod{e_{0}}{p}_w}$.

Armed as such, we substitute in \cref{eq:1} to get that we need to show that:
\begin{gather*}
\sum_{i = 1}^t \frac{b_i - b_{i + 1}}{2} \cdot \sgn{\innerprod{e_i}{p}_w} \geq 0 \iff \\ 
\sum_{i = 1}^t (b_i - b_{i + 1}) \cdot \sgn{\innerprod{e_i}{p}_w} \geq 0 \iff \\
\sum_{i = 1}^{t} b_i \cdot \sgn{\innerprod{e_i}{p}_w} -
\sum_{i = 1}^{t} b_{i + 1} \cdot \sgn{\innerprod{e_i}{p}_w} \geq 0 
\end{gather*}
The last term in the second sum is $b_{t + 1} \cdot \sgn{\innerprod{e_t}{p}_w} = b_1 \cdot \sgn{\innerprod{e_{0}}{p}_w},$ so the second sum stays the same if we change its summation bounds from $(1, t)$ to $(0, t - 1)$. Doing so and then rewriting in terms of $i + 1$ instead of $i$, the second sum equals $\sum_{i = 1}^{t} b_{i} \cdot \sgn{\innerprod{e_{i - 1}}{p}_w}$. Combining the two sums and then factoring out the $b_i$ yields: 
\begin{gather*}
\sum_{i = 1}^{t} b_i \cdot (\sgn{\innerprod{e_i}{p}_w} - \sgn{\innerprod{e_{i - 1}}{p}_w}) \geq 0
\end{gather*}
To show that this is true, we will just show that each term is non-negative. Consider a fixed $i \in [t]$, then we would like to show that $b_i \cdot (\sgn{\innerprod{e_i}{p}_w} - \sgn{\innerprod{e_{i - 1}}{p}_w}) \geq 0$. This happens if and only if:
\begin{gather*}
    \sgn{b_i} \cdot (\sgn{\innerprod{e_i}{p}_w} - \sgn{\innerprod{e_{i - 1}}{p}_w}) \geq 0 \iff \\
    \sgn{b_i} \cdot \sgn{\innerprod{e_i}{p}_w} \geq \sgn{b_i} \cdot \sgn{\innerprod{e_{i - 1}}{p}_w} \iff \\
    \sgn{b_i \cdot \innerprod{e_i}{p}_w} \geq \sgn{b_i \cdot \innerprod{e_{i - 1}}{p}_w} 
\end{gather*}
Because the sign function is monotonic, it hence suffices to show that $b_i \cdot \innerprod{e_i}{p}_w \geq b_i \cdot \innerprod{e_{i - 1}}{p}_w \iff b_i \cdot \innerprod{e_i - e_{i - 1}}{p}_w \geq 0 \iff b_i \cdot (2 \cdot w_i \cdot p_i) \geq 0$, which is true since $b_i \cdot p_i \geq 0$.
\end{proof}

The general statement that Ostrogorski's paradox does not occur will now follow easily by combining \cref{lemma:anscombe-does-not-occur-for-single-switch} with the following:

\begin{lemma}\label{lem:ansc_ost_relation} Ostrogroski's paradox occurs for an instance $\votinginstance = (\prefprof, w)$ in the external-weights model if and only if Anscombe's paradox occurs on an instance $\votinginstance'$ obtained from $\votinginstance$ by removing (possibly zero) issues from $\votinginstance$ and renormalizing the weights to sum up to 1.
\end{lemma}
\begin{proof} We prove the two directions separately:

$(\Leftarrow)$ Assume Anscombe's paradox occurs on an instance $\votinginstance'$ obtained from $\votinginstance$ by removing a subset of issues $R \subseteq [t]$ from $\votinginstance$ and renormalizing the weights to sum up to 1. Then, by definition, there is an IWM proposal $p_{IWM}'$ for $\votinginstance'$ such that $\overline{p_{IWM}'} \succ_{\votinginstance'} p_{IWM}'$. Complete $p_{IWM}'$ into an IWM proposal $p_{IWM}$ for $\votinginstance$ and define $p_*$ to agree with $p_{IWM}$ in topics in $R$ and disagree in topics in $[t] \setminus R$. Then, Ostrogorski's paradox occurs for $\votinginstance$: proposal $p_{IWM}$ is an IWM and $p_* \succ_{\votinginstance} p_{IWM}$ by construction because $\overline{p_{IWM}'} \succ_{\votinginstance'} p_{IWM}'$.

$(\Rightarrow)$ Assume Ostrogroski's paradox occurs for $\votinginstance$. Let $p_{IWM}$ and $p_*$ be such that $p_{IWM}$ is an IWM proposal for $\votinginstance$ and $p_* \succ_{\votinginstance} p_{IWM}$. Define $R \subseteq [t]$ to be the set of topics in which $p_*$ and $p_{IWM}$ agree, and create $\votinginstance'$ from $\votinginstance$ by removing issues in $R$ and renormalizing the weights to sum up to 1. Moreover, restrict $p_{IWM}$ to topics in $[t] \setminus R$ to get an IWM proposal $p_{IWM}'$ for $\votinginstance'$. Then, Anscombe's paradox occurs for $\votinginstance'$: proposal $p_{IWM}'$ is an IWM and $\overline{p_{IWM}'} \succ_{\votinginstance'} p_{IWM}'$ by construction because $p_* \succ_{\votinginstance} p_{IWM}$.
\end{proof}

\externalswitchcondorcet*
\begin{proof} Assume this was not the case and consider a single-switch instance $\votinginstance$ in the external-weights model for which there exists an issue-wise majority proposal that is not a Condorcet winner (equivalently, Ostrogorski's paradox occurs for $\votinginstance$). By \cref{lem:ansc_ost_relation}, Anscombe's paradox occurs on an instance $\votinginstance'$
obtained from $\votinginstance$ by removing (possibly zero) issues from $\votinginstance$ and renormalizing the weights to sum up to 1. Since $\votinginstance$ is single-switch, so is $\votinginstance'$. Hence, Anscombe's paradox occurs in a single-switch instance, contradicting \cref{lemma:anscombe-does-not-occur-for-single-switch}.
\end{proof}

\subsection{Recognizing Single-Switch Profiles}\label{app:reco-ssw}

In this section, we first prove \cref{lemma:cyclic-shift-and-negate}, restated below for convenience. Then, we delve into the task of recognizing single-switch-no-flips profiles, providing an ample guided discussion of the relation between our simpler $O(nt)$ algorithm and related work, including the task of recognizing single-crossing preferences. As a bonus, we discuss a simpler and, at the same time, more efficient algorithm for recognizing single-crossing preferences, running in time $O(nt \sqrt{\log n})$.

\lemmacyclicshift*
\begin{proof} It is enough to show this for the case where $\prefprof$ has $n = 1$ rows, in which case $c_1, \ldots, c_t \in \Bset$. Let us assume $c_1 = 1$, as the other case is analogous. Since $\prefprof'$ is an SSW presentation of $\prefprof$, let $1 \leq k \leq t$ be such that $\prefprof'$ starts with $k$ ones and the rest are $-1$'s. As a result, $\prefprof'_r$ by definition starts with $k - 1$ ones, and the rest are $-1$'s, implying that $\prefprof'_r$ is an SSW presentation of $\prefprof'$, and hence also of $\prefprof$ by transitivity. 

To get the part about $\prefprof''$, apply the previous reasoning for $\prefprof'_r$ repeatedly, each time taking the first column, negating it, and moving it to the end. Doing so $2t$ times leads back to the original presentation, and along the way, we get the advertised single-switch presentations.
\end{proof}

\textbf{Recognizing single-switch-no-flips profiles.} We now focus on deciding whether a given preference profile $\prefprof$ is single-switch-no-flips. A first way to do so requires rather involved machinery, by reducing to the \emph{Consecutive Ones Problem} (C1P). In the C1P problem, the input is an $n \times t$ matrix with $\pm 1$ entries. The goal is to permute its columns so that the ones on each row form a consecutive interval. Solving C1P on a matrix with a negated copy of itself appended underneath corresponds to requiring a solution for the original matrix where not only ones are consecutive, but also $-1$'s, meaning that ones form a prefix or a suffix on each row. C1P can be solved in $O(n t)$ time \cite{booth_lueker}, hence giving an immediate solution to check whether $\prefprof$ is single-switch-no-flips within the same time. However, linear-time C1P solvers are complicated and notoriously error-prone: most available implementations fail on at least some edge cases \cite{experimental_pq_trees}. Moreover, reducing to C1P does not utilize the additional structure present in our problem and hence does not shed light on the structure of all solutions, as we set out to do. We note that this way of checking whether a profile is single-switch-no-flips has previously appeared in \cite{elkind_lackner_dichotomous}, where the authors use it to solve the equivalent problem of recognizing voter/candidate-extremal-interval preferences.

A second way to recognize single-switch-no-flips profiles reduces to the problem of recognizing \emph{single-crossing preferences}. In this problem, the input consists of a set $A$, of size denoted by $n$, and a list $(\succ_1, \ldots, \succ_t)$ of linear orders over $A$. The goal is to permute this list to obtain a new list $(\succ_1', \ldots, \succ_t')$ such that for any $a, a' \in A$ with $a \neq a'$ the set $\{j \in [t] : a \succ_j' a'\}$ forms a prefix or a suffix. The reduction is not too difficult: start with a preference profile $\prefprof$ and define $A = \cup_{i \in [n]}\{a_i^0, a_i^1\}$ and a list $(\succ_1, \ldots, \succ_t)$ of linear orders over $A$ as follows. For each $j \in [t]$, order $\succ_j$ ranks the elements in $A$ as $\{a_1^0, a_1^1\} \succ_j \{a_2^0, a_2^1\} \succ_j \ldots \succ_j \{a_n^0, a_n^1\}$, breaking the tie inside each bracket as follows: for each $i \in [n]$, rank $a_i^0\succ_j a_i^1$ if $v_{i, j} = 1$, and $a_i^1\succ_j a_i^0$ otherwise. One can check that permutations of the list with the required property correspond to permutations of the columns of $\prefprof$ such that the ones on each row form a prefix or a suffix. Recognizing single-crossing preferences can be achieved in $O(nt \log n))$ time \cite[Algorithm 4]{survey_restricted}, meaning that our problem also can. In contrast to the C1P approach, this algorithm has a reasonable implementation. Moreover, by the standard fact that, when it exists, the single-crossing permutation is unique up to reversal, we get that the SSWNF presentation is unique up to reversal whenever it exists (we will give a self-contained proof later on, so we omit the details here for brevity).

The previous approach can be modified to run in time $O(nt)$ by identifying and adapting its super-linear components. Most prominently, \cite[Algorithm 4]{survey_restricted} computes $O(t)$ times the Kendall Tau distance between certain orders $\succ_i$ and $\succ_j$ with $i, j \in [t]$, which is defined as the number of pairs of elements on which $\succ_i$ and $\succ_j$ disagree: $d_{KT}(\succ_i, \succ_j) := \{(a, a') \in A^2 : a \succ_i a'\text{ and } a' \succ_i a\}|$. This is done in time $O(n \log n)$ by finding the number of inversions of a permutation. However, for our particular construction of the orders $(\succ_1, \ldots, \succ_t)$, disagreements between orders can only occur on pairs of the form $(a_k^0, a_k^1)$ with $k \in [n]$, and the number of such disagreements is precisely the Hamming distance $d_H(c_i, c_j)$, which can be computed in time $O(n)$. The algorithm also sorts a list of $O(t)$ integers with values bounded by $O(n^2)$. To get the right time complexity, this is done depending on whether $n$ or $t$ is larger, either in time $O(t \log t)$ or using Counting Sort in time $O(n^2 + t)$. We note that using Radix Sort would have sufficed to make this $O(n + t)$ without the need for a case distinction.\footnote{With this small modification, and if counting inversions is performed with the $O(n \sqrt{\log n})$ algorithm of \cite{patrascu_inversions}, single-crossing preferences can be recognized in the better time $O(n t \sqrt{\log n})$. We do not give the details here in order not to dilute the message. Instead, we discuss a simpler $O(n t \sqrt{\log n})$ algorithm at the end of the section based on similar ideas.} For our usage, the values in the list are instead bounded by $O(n)$, so Counting Sort suffices directly (this is because Kendall Tau distances can be quadratic in $n$, while Hamming distances only linear in $n$). This completes the required modifications. We note, moreover, that their algorithm proceeds in two stages: first, a candidate list $(\succ_1', \ldots, \succ_t')$ is determined, and then it is checked whether it satisfies the single-crossing condition. Moreover, should a solution exist, the candidate list is the unique one up to reversing the list. The second stage is, perhaps surprisingly, the more difficult one to achieve efficiently, and its correctness proof is the subtler part of the argument. For our purposes, however, the first stage suffices since, given the candidate solution, it is easy to check whether ones on each row form a prefix or a suffix in additional time $O(nt)$.

It is possible to give a self-contained $O(nt)$ algorithm for our problem following the outline above (without going through the reduction to single-crossing preferences). However, the resulting algorithm is still arguably not the simplest. Instead, in the following we present a simpler, $O(nt)$ direct algorithm for recognizing single-switch-no-flips profiles $\prefprof = (c_1, \ldots, c_t)$. The algorithm combines insights from \cite[Section 4.2]{survey_restricted} with a simple new observation. We defer further elaboration on the connection with single-crossing preferences until the end of the section for clarity.
Our algorithm proceeds as follows: First, find an index $x$ maximizing $d_H(c_1, c_x)$. Then, sort (using Counting Sort) the columns based on their Hamming distance from $c_x$ to get a profile $\prefprof' = (c_1', \ldots, c_t')$ where $d_H(c_x, c_i') \leq d_H(c_x, c_{i + 1}')$ for $i \in [t - 1]$ (i.e., ties in Hamming distance can be broken arbitrarily). We claim that either $\prefprof'$ is the unique SSWNF presentation of $\prefprof$ (up to reversing the order of the columns), or there is no such presentation. The algorithm runs in $O(nt)$ time. We prove all required claims in the following theorem:

\recosswnf*

\begin{proof} We first repeat the full algorithm briefly: 
\begin{enumerate}
\item Say $\prefprof = (c_1, \ldots, c_t)$;
\item Find an index $x$ maximizing $d_H(c_1, c_x)$;
\item Sort (using Counting Sort) the columns based on their Hamming distance from $c_x$ to get a profile $\prefprof' = (c_1', \ldots, c_t')$ where $d_H(c_x, c_i') \leq d_H(c_x, c_{i + 1}')$ for $i \in [t - 1]$; 
\item Check whether ones form a prefix or a suffix on each row of $\prefprof'$. If yes, return $\prefprof'$, else return ``not single-switch-no-flips.''
\end{enumerate}

Each step of the algorithm runs in time $O(nt)$, so it attains the required time-bound. It remains to show that: (i) if no solution exists, then this is correctly reported; (ii) if $\prefprof^* = (c_1^*, \ldots, c_t^*)$ is an arbitrary SSWNF presentation of $\prefprof$, then the algorithm returns either $\prefprof^*$ or its reverse (which note further implies that the solution is unique up to reversal). Part (i) follows immediately from the last step in the algorithm. To show part (ii), we will need an auxiliary claim:

\begin{claim}\label{claim:hamming-monot} Assume $1 \leq i < j < k \leq t$, then $d_H(c_i^*, c_j^*) \leq d_H(c_i^*, c_k^*)$ with equality iff $c_j^* = c_k^*$. Similarly, $d_H(c_i^*, c_k^*) \leq d_H(c_j^*, c_k^*)$ with equality iff $c_i^* = c_j^*$. 
\end{claim}
\begin{proof} We will only prove the first part, as the second follows analogously. First, observe that for any $u, v \in \Bset^n$ we can write $d_H(u, v) = \sum_{\ell = 1}^n \mathbb{I}(u_\ell \neq v_\ell)$. Hence, it suffices to show that for each $1 \leq \ell \leq n$ we have $\mathbb{I}(c_{i, \ell}^* \neq c_{j, \ell}^*) \leq \mathbb{I}(c_{i, \ell}^* \neq c_{k, \ell}^*)$ and then sum up those inequalities. This amounts to showing that $c_{i, \ell}^* \neq c_{j, \ell}^* \implies c_{i, \ell}^* \neq c_{k, \ell}^*$. Assume the contrary: $c_{k, \ell}^* = c_{i, \ell}^* \neq c_{j, \ell}^*$. This means that on row $\ell$ of $\prefprof^*$ columns $i < j < k$ read either $+1, -1, +1$ or $-1, +1, -1$, which either way means that ones do not form a prefix or a suffix on this row, contradicting that $\prefprof^*$ is an SSWNF presentation of $\prefprof$. To get the equality case, note that we have shown $d_H(c_i^*, c_j^*) \leq d_H(c_i^*, c_k^*)$ by summing inequalities, so equality occurs iff all summed inequalities are tight; i.e., for all $1 \leq \ell \leq n$ we have $\mathbb{I}(c_{i, \ell}^* \neq c_{j, \ell}^*) = \mathbb{I}(c_{i, \ell}^* \neq c_{k, \ell}^*)$, which amounts to $c_{i, \ell}^* \neq c_{j, \ell}^* \iff c_{i, \ell}^* \neq c_{k, \ell}^*$, and in turn to $c_{i, \ell}^* = c_{j, \ell}^* \iff c_{i, \ell}^* = c_{k, \ell}^*$, which holds iff $c_{j, \ell}^* = c_{k, \ell}^*$. As a result, equality occurs iff for all $1 \leq \ell \leq n$ we have $c_{j, \ell}^* = c_{k, \ell}^*$, i.e., $c_j^* = c_k^*$.
\end{proof}

Armed as such, we can show that the column $c_x$ maximizing $d_H(c_1, c_x)$ found by the algorithm satisfies $c_x \in \{c_1^*, c_t^*\}$. To see this, let $y$ be such that $c_1 = c_y^*$, then any column $c_x$ that maximizes $d_H(c_1, c_x)$ equals a column $c_z^*$ that maximizes $d_H(c_y^*, c_z^*)$. By \cref{claim:hamming-monot}, it follows that $c_z \in \{c_1^*, c_t^*\}$, so $c_x \in \{c_1^*, c_t^*\}$.

In the following, we will show that if $c_x = c_1^*$, then the output of the algorithm is $\prefprof' = \prefprof^*$, and if $c_x = c_t^*$, then $\prefprof' = \text{the reverse of }\prefprof^*$, completing the proof.

The equality parts of \cref{claim:hamming-monot} give that two columns in $\prefprof^*$ have the same Hamming distance from $c_1^*$ (or $c_t^*$), if and only if they are equal. Since $c_x \in \{c_1^*, c_t^*\}$, this means that two columns in $\prefprof$ have the same Hamming distance from $c_x$ if and only if they are equal.\footnote{Note that this hinges on our assumption that $\prefprof$ admits $\prefprof^*$ as an SSWNF presentation (and would be false in general).}

The algorithm constructs $\prefprof'$ by ordering the columns of $\prefprof$ in non-decreasing order of Hamming distance from $c_x$. By the previous, equal Hamming distances correspond to identical columns, so $\prefprof'$ as defined by us to satisfy $d_H(c_x, c_i') \leq d_H(c_x, c_{i + 1}')$ for $i \in [t - 1]$ is actually unique no matter the tie-breaking for equal distances.\footnote{Fact which again is only true because in this part of the proof we assumed that $\prefprof$ is single-switch-no-flips.}

If $c_x = c_1^*$, note that $(c_1', \ldots, c_t') = (c_1^*, \ldots, c_t^*)$ satisfies $d_H(c_x, c_i') \leq d_H(c_x, c_{i + 1}')$ for $i \in [t - 1]$. By the uniqueness of $(c_1', \ldots, c_t')$, this means that $(c_1^*, \ldots, c_t^*)$ is the output of our algorithm, i.e., $\prefprof' = \prefprof^*$.

If $c_x = c_t^*$, the reasoning is analogous, leading to $\prefprof' =$ the reverse of $\prefprof^*$.
\end{proof}

\textbf{Better recognition for single-crossing preferences.} Our $O(nt)$ algorithm for recognizing single-switch-no-flips profiles can be easily modified to recognize single-crossing preferences: Consider an input consisting of a set $A$ and a list $(\succ_1, \ldots, \succ_t)$ of linear orders over $A$. First, find an index $x$ maximizing $d_{KT}(\succ_1, \succ_x)$. Then, sort (using Radix Sort) the list according to the Kendall Tau distance from $\succ_x$ to get a list $(\succ_1', \ldots, \succ_t')$ where $d_{KT}(\succ_x, \succ_i') \leq d_{KT}(\succ_x, \succ_{i + 1}')$ for $i \in [t - 1]$ (i.e., ties can be broken arbitrarily). The resulting list is either the unique permutation of the input list witnessing the single-crossing property (up to reversal), or there is no such permutation. After identifying this candidate solution, one checks whether the single-crossing property is satisfied using the second stage of \cite[Algorithm 4]{survey_restricted}, namely \cite[Algorithm 2]{survey_restricted}: the solution is valid if and only if $d_{KT}(\succ_1', \succ_i') + d_{KT}(\succ_i', \succ_{i + 1}') = d_{KT}(\succ_1', \succ_{i + 1}')$ for all $2 \leq i < t$. If the $O(n \sqrt{\log n})$ algorithm of \cite{patrascu_inversions} is used for counting permutation inversions, the previous yields a simpler, and at the same time, more efficient algorithm for recognizing single-crossing preferences, running in time $O(nt \sqrt{\log n})$. The proof of correctness of this algorithm follows the same outline as the proof of \cref{th:sswnf-reco}.\footnote{This is no accident: one can reduce from recognizing single-crossing preferences to recognizing single-switch-no-flips profiles by introducing a voter for each pair of distinct elements in $A$. This incurs a quadratic computational cost but suffices to recover correctness.}
We also note that our algorithm is, in fact, a more efficient implementation of \cite[Algorithm 3]{survey_restricted}. The latter tries all $O(t)$ options for $\succ_1'$ and, for each one, sorts by Kendall Tau distance from $\succ_1'$ to get a candidate solution, which is then verified as in our case using \cite[Algorithm 2]{survey_restricted}. Our improvement was to notice that $\succ_1'$ can be determined efficiently without trying out all options for it, hence removing a linear factor from the time complexity. In contrast, instead of efficiently determining $\succ_1'$, \cite[Algorithm 4]{survey_restricted} takes a more intricate approach.

\begin{theorem} \label{th:fast-sc-rec} Single-crossing preferences can be recognized in time $O(nt \sqrt{\log n})$, including producing a witnessing permutation for yes-instances (which is the unique solution up to reversal).
\end{theorem}

\subsection{Forbidden Subprofiles Characterization of Single-Switch Preferences}\label{app:forbidden}

In this section, we prove \cref{th:forbidden-ssw-main}, which establishes the forbidden subprofiles characterization of single-switch profiles. Afterward, we give further details about how our recognition algorithm for single-crossing preferences in \cref{th:fast-sc-rec} can be bootstrapped to also produce a forbidden subinstance without sacrificing runtime, similarly to the proof of \cref{th:fast-forbidden-subprofiles-ssw}.

To begin with proving \cref{th:forbidden-ssw-main}, a profile $\prefprof$ is single-switch if and only if $\prefprof'$ is single-switch-no-flips, where $\prefprof'$ is the profile obtained from $\prefprof$ by flipping columns so that the first row is all $-1$'s. As a result, it suffices to understand how short proofs of non-membership look for the class of single-switch-no-flips preferences. To this end, \cite{Terzopoulou_Karpov_Obraztsova_2021} considered the profiles:
\[
\prefprof_1 =\begin{bmatrix}
\tikzmarknode{p111}{\texttt{+1}} & \tikzmarknode{p112}{\texttt{+1}} & \tikzmarknode{p113}{\texttt{-1}} & \tikzmarknode{p114}{\texttt{-1}} \\
\tikzmarknode{p121}{\texttt{+1}} & \tikzmarknode{p122}{\texttt{-1}} & \tikzmarknode{p123}{\texttt{+1}} & \tikzmarknode{p124}{\texttt{-1}}
\end{bmatrix}\quad
\prefprof_2 = \begin{bmatrix}
\tikzmarknode{p211}{\texttt{+1}} & \tikzmarknode{p212}{\texttt{-1}} & \tikzmarknode{p213}{\texttt{-1}} \\
\tikzmarknode{p221}{\texttt{-1}} & \tikzmarknode{p222}{\texttt{+1}} & \tikzmarknode{p223}{\texttt{-1}} \\
\tikzmarknode{p231}{\texttt{-1}} & \tikzmarknode{p232}{\texttt{-1}} & \tikzmarknode{p233}{\texttt{+1}} 
\end{bmatrix}\quad
\]%
\begin{tikzpicture}[overlay,remember picture, shorten >=-3pt, shorten <= -3pt]
\drawcell{p111}{\cone}
\drawcell{p112}{\cone}
\drawcell{p113}{\czero}
\drawcell{p114}{\czero}
\drawcell{p121}{\cone}
\drawcell{p122}{\czero}
\drawcell{p123}{\cone}
\drawcell{p124}{\czero}
\drawcell{p211}{\cone}
\drawcell{p212}{\czero}
\drawcell{p213}{\czero}
\drawcell{p221}{\czero}
\drawcell{p222}{\cone}
\drawcell{p223}{\czero}
\drawcell{p231}{\czero}
\drawcell{p232}{\czero}
\drawcell{p233}{\cone}
\end{tikzpicture}%
and showed that $\prefprof'$ is single-switch-no-flips if and only if it does not contain as subprofiles $\prefprof_1$, $\prefprof_2$ and the profiles obtainable from them by flipping any subset of rows. In total, one can check that this leads to 5 non-equivalent profiles $\prefprof_1, \ldots, \prefprof_5$, where $\prefprof_1, \prefprof_2$ are the ones above and $\prefprof_{2 + i}$ for $1 \leq i \leq 3$ is $\prefprof_2$ with the first $i$ rows flipped. The original presentation lists the 5 profiles explicitly, but since both single-switch and single-switch-no-flips preferences are closed under flipping rows, we find our account cleaner.
For the interested reader, we note that the result of \cite{Terzopoulou_Karpov_Obraztsova_2021} can also be recovered from \cite{bredereck_sc_forbidden_minors}, where it is shown that non-membership to the class of single-crossing preferences is always witnessed by one of two small subinstances (this can be done using as a lens the reduction to single-crossing preferences recognition in \cref{app:reco-ssw}).

\begin{lemma} A profile $\prefprof$ is single-switch if and only if $\prefprof'$ does not contain as a subprofile any of $\prefprof_1, \ldots, \prefprof_5$. Here $\prefprof'$ denotes the profile obtained from $\prefprof$ by flipping columns so that the first row is all $-1$'s.
\end{lemma}

This already gives short proofs for non-membership, but for our purposes, we would like a characterization 
in terms of the subprofiles of $\prefprof$, not of $\prefprof'$. This can be easily achieved given the previous: consider a non-single-switch profile $\prefprof$, then, by the previous, $\prefprof'$ contains one of $\prefprof_1, \ldots, \prefprof_5$, say $\prefprof_i$. Note that $\prefprof_i$ has no row that is all $-1$'s, so $\prefprof_i$ is, in fact, a subprofile of $\prefprof'$ without its first row (which is all $-1$'s). As a result, if we define $\prefprof_i^a$ to be $\prefprof_i$ with a row of $-1$'s appended to the top, then $\prefprof'$ also contains $\prefprof_i^a$ as a subprofile. Given how $\prefprof'$ was obtained from $\prefprof$ by flipping a subset of columns, $\prefprof$ contains a version of $\prefprof_i^a$ with accordingly-flipped columns. Namely, if we let $\Pi_i^a$ be the set of profiles that can be obtained from $\prefprof_i^a$ by flipping any subset of columns, then $\prefprof$ contains some $\mathcal{X} \in \Pi_i^a$ as a subprofile. Moreover, $\mathcal{X}$ is not single-switch, as $\mathcal{X}' = \prefprof_i^a$ and $\prefprof_i^a$ is not single-switch-no-flips (here $\mathcal{X}'$ denotes $\mathcal{X}$ with columns flipped to make its first row all $-1$'s). As a result, we get the following:

\begin{lemma} A profile $\prefprof$ is single-switch if and only if it contains no profile $\mathcal{X} \in \cup_{i = 1}^5 \Pi_i^a$ as a subprofile.
\end{lemma}

This result can be compressed into a more elegant form by leveraging the closure of single-switch preferences under flipping rows and columns and the way $\prefprof_3, \ldots, \prefprof_5$ were obtained from $\prefprof_2$, as follows, which is precisely \cref{th:forbidden-ssw-main}, restated below for convenience:

\forbiddensswmain*

\textbf{Finding forbidden subinstances of \emph{single-crossing} preferences.} Our idea to use a fast black-box recognition algorithm to bootstrap a fast algorithm for finding a forbidden subprofile extends beyond our usage for single-switch preferences in 
\cref{th:fast-forbidden-subprofiles-ssw}. The same idea can be used for the class of single-crossing preferences, where non-membership to the class is witnessed by one of two small \emph{forbidden subinstances} with $(t, n) \in \{(4, 4), (3, 6)\}$, as shown in \cite{bredereck_sc_forbidden_minors}. Given our $O(nt\sqrt{\log{n}})$ recognition algorithm in \cref{th:fast-sc-rec}, we can apply very similar reasoning to the proof of 
\cref{th:fast-forbidden-subprofiles-ssw} (except now we need to split into $6 + 1 = 7$ groups) to get the following result, which, to the best of our knowledge, is new, even if we were to replace our improved time bound with that of the previously-known fastest algorithm, namely \cite[Algorithm 4]{survey_restricted}:

\begin{theorem} Given a no-instance of the problem of recognizing single-crossing preferences, a forbidden subinstance can be determined in time $O(nt\sqrt{\log{n}})$.
\end{theorem}

\section{Anscombe's Paradox}

\subsection{External Weights}\label{app:external_iwm}

This section of the appendix centers on \cref{thm:external_iwm}. In \cref{app:relevant_topics}, we formalize the definition of relevant topics and discuss how to determine the set of relevant topics for a given instance $\mathcal{I} = (\mathcal{P}, w)$ efficiently. \cref{app:existence_external_iwm} then proves \cref{thm:external_iwm}. Throughout this section of the appendix, we use the notation $\w(S) = \sum_{j \in S} \w_j$ for any $S \subseteq [t]$.

\subsubsection{Relevant Topics}\label{app:relevant_topics}

 A subset of the topics, $T \subseteq [t]$, is a \textit{minimal topic group} under weight vector $\w$ if $\w(T) > \frac{1}{2}$ and for all $j \in T$ we have that $\w(T \setminus \{j\}) \leq \frac{1}{2}$. Then we call a topic $j \in [t]$ a \textit{relevant topic} under $\w$ if it is in some minimal topic group under $\w$. We denote the set of all relevant topics under $\w$ as $R_{\w}$.

In proving \cref{thm:te1} (see \cref{app:existence_external_iwm}), we define $B_m = \{p \in \Bset^t \colon \langle p, p_{IWM}\rangle_w > 0\}$ and select a proposal from it uniformly at random. Note that we can equivalently write $B_m = \{p \in \Bset^t \colon d_H(p, p_{IWM}, w) < 1/2\}$. We then claim that the only topics $j \in [t]$ such that $\Pr(p_j = (p_{IWM})_j) - \Pr(p_j \neq (p_{IWM})_j) > 0$ are relevant topics. We prove this claim below.\\

\begin{claim}\label{claim:relevant_probs}
For $p$ selected uniformly at random from $B_m$, $\Pr(p_j = (p_{IWM})_j) > \frac{1}{2}$ if and only if $j \in R_{\w}$.
\end{claim}

\begin{proof}
    Fix some $j \in [t]$. We denote the number of proposals with $+(p_{IWM})_j$ for $j$ and $-(p_{IWM})_j$ for $j$ in $B_m$ by $N_+, N_-$ respectively.\\
    
     $(\Leftarrow)$ Assume $j \in R_w$. Consider the bijection $q: \Bset^t \rightarrow \Bset^t$ that flips the $j$th entry of the proposal. If $d_H(p, p_{IWM}, w) < 1/2 - w_j$ and $p_j = (p_{IWM})_j$ then $p \in B_m$ and so is $q(p)$, so these cancel each other out when comparing $N_+ - N_-$. Note too that if $p_j = -(p_{IWM})_j$ and $p \in B_m$ then certainly $q(p)$ is as well. So we only have to consider the case where $p = (p_{IWM})_j, p \in B_m$ but $d_H(p, p_{IWM}, w) \geq 1/2 - w_j$. We know such a case must exist because $j \in R_w$. So, this means $j$ is in some minimal topic group, $T$. Let $p$ be the proposal where all topics in $T$ are set to their value in $p_{IWM}$, and all remaining topics are set to their value in $-p_{IWM}$. By definition of minimal topic group, we have that $w(T) \leq 1/2 + w_j$, so $d_H(p, p_{IWM}, w) = 1 - w(T) \geq 1/2 - w_j$. Then $d_H(q(p), p_{IWM}, w) \geq 1/2$, so $q(p) \not \in B_m$. So we have that $N_+ > N_-$. Hence $\Pr(p_j = (p_{IWM})_j) > 1/2$.\\

    $(\Rightarrow)$ Assume $\Pr(p_j = ((p_{IWM})_j) > 1/2$. We know that all $p \in B_m$ with $p_j = -(p_{IWM})_j$ have $q(p) \in B_m$, so this means there must exist some $p \not \in B_m$ with $p_j = -(p_{IWM})_j$ such that $q(p) \in B_m$. Find the $p$ with the fewest indices agreeing with $p_{IWM}$ such that this holds. Then the set of topics that $q(p)$ agrees with $p_{IWM}$ on forms a minimal topic group for $\w$. To see this, let $T$ be the set of topics that $q(p)$ agrees with $p_{IWM}$ on. We know that $w(T) > 1/2$ and $w(T \setminus \{j\}) \leq 1/2$ because $p \not \in B_m, q(p) \in B_m$. Take any other $k \in T, k \neq j$. Assume for sake of contradiction that $W(T \setminus {k}) > 1/2$. Then we could have taken $p' \not \in B_m$ to be equal to $p$ at all indices except flipped for $k$. Then $p' \not \in B_m$ and $q(p') \in B_m$, but $p'$ has one fewer index agreeing with $p_{IWM}$ than $p$. This contradicts the minimality in our selection of $p$. Hence $T$ is a valid minimal topic group, and thus $j \in R_{\w}$.
\end{proof}

As \cref{thm:external_iwm} shows, relevant topics are useful for determining whether a proposal with majority support and with weighted agreement of $> 1/2$ with the issue-wise majority exists. How, then can we discern whether specific topics are relevant? This is essentially equivalent to determining whether or not a voter is a dummy voter in a weighted majority game. While this can be shown to be NP-hard with respect to general input $w$~\cite{matsui2000survey}, we have the constraint that all elements of $w$ are bounded in $[0, 1]$. We say that $w$ has \emph{polynomial precision} if all of its elements can be expressed as rational numbers with a common denominator that is polynomial in $n$ and $t$. If additionally we assume that $w$ has polynomial precision then we can give a polynomial time algorithm for determining the set of relevant topics. In order to do this we show the following claim:\\

\begin{claim}[Relevance Monotonicity]\label{claim:rel_mono}
    Being a relevant topic is a monotonic property with respect to topic weight.
\end{claim}

\begin{proof}
    Fix topics $i,j$ such that $w_j \geq w_i$. We want to show that if $i$ is a relevant topic then $j$ is also a relevant topic. Assume $i$ is relevant, then there exists some $S \subseteq [t] \setminus \{i\}$ such that $w(S) \leq 1/2$ but $w(S \cup \{i\}) > 1/2$. We have two cases:\begin{enumerate}
        \item $\mathbf{j \not \in S}.$ In this case we can reuse $S$ as proof of $j$'s relevance. We have that $w(S) \leq 1/2$ and $w(S \cup \{j\}) = w(S) + w(j) \geq w(S) + w(i) = w(S \cup \{i\}) > 1/2$.
        \item $\mathbf{j \in S}.$ In this case we define $S' = (S \setminus \{j\}) \cup \{i\}$ (we just swap in $i$ for $j$). Then we have that $w(S') = w(S) - w(j) + w(i) \leq w(S) \leq 1/2$. We also can see that $w(S' \cup \{j\}) = w(S \cup \{i\}) > 1/2$.
    \end{enumerate}
    Therefore, we have shown that if $i$ is relevant, then $j$ is also relevant. This is equivalent to saying that if $j$ is irrelevant, $i$ is irrelevant as well.
\end{proof}

\cref{claim:rel_mono} implies that there is a ``lowest weight relevant topic'' such that all topics with weight less than it are irrelevant, and all topics with greater weight must be relevant. In order to find the lowest weight relevant topic we can run binary search over the topics sorted by weight. For each topic choice $j$, we can run knapsack on the remaining topics to look for a set such that $\frac{1}{2} - w_j < w(S) \leq 1/2$ as proof of its relevance. Assuming that all weights are integral, which we can achieve by temporarily scaling all of the weights and the respective bounds up, then there is an exact algorithm for knapsack that runs in polynomial time with respect to $|w| = t$ and the max weight item. As we only have to scale up the weights a polynomial factor with respect to $nt$ (due to our assumption on their precision), the max weight item is also polynomial in $nt$. Therefore, the whole algorithm together is polynomial in $nt$. As any preference profile input has size at least $n \times t$ (just considering number of entries), this is a polynomial size operation with respect to the size of our overall problem input.

\subsubsection{Existence of Representative Non-Losing Proposals in the External Weights Setting}\label{app:existence_external_iwm}

In this section, we prove \cref{thm:external_iwm}, restated here:
\externaliwm*

At a high level, our proof structure is the following (the same structure as used in~\cite{constantinescu2023computing}): For every voter, we define two bijective relations between proposals. Then, we piece these together to construct a third relation that ``swaps'' a proposal's weighted similarity to the voter with its weighted similarity to $p_{IWM}$. We then define two quantities that we take the expectation of. One can easily be shown to be non-negative, and we use the defined relations for each voter to show the equality of these expectations. The non-negativity of the second expectation then implies the existence of a majority-supported proposal, $p \in \Bset^t$, with $\hamdist(p_{IWM}, p, w) < 1/2$.

\subsubsection{Structure-Preserving Maps}
We define three different proposal transformations for each given voter. Fix some voter $v$ and let $B^* \subseteq \mathbb{B}^t$ be the set of proposals $p$ such that $\langle p_{IWM}, p \rangle_w \neq 0$ and $\langle v, p \rangle_w \neq 0$. We partition proposals in $B^*$ into four categories: $T_{i, j}$ where $i, j \in \{-1, +1\}$ and $i = sgn(\langle p, p_{IWM} \rangle_w)$ and $j = sgn(\langle v, p \rangle_w)$. We construct the following ``mask'' for use in our transformations: $v \odot p_{IWM}$ where $\odot$ is the elementwise (Hadamard) product. This mask has $+1$ for topics on which $v$ and $p_{IWM}$ agree, and $-1$ for topics on which they disagree. Then we define two transformations $f_v^+, f_v^- \colon B^* \rightarrow B^*$ as follows: \[
    f_v^+(p) = p \odot (v \odot p_{IWM}),\ \  f_v^-(p) = p \odot -(v \odot p_{IWM})
\] 
$f_v^+$ flips a proposal on indices where $v$ and $p_{IWM}$ disagree, while $f_v^-$ flips a proposal on indices where $v$ and $p_{IWM}$ agree. Note that the Hadamard product is commutative and for any vector $m \in \mathbb{B}^t, m \odot m = \mathbf{1}$. Additionally, for any vectors $a, b, c \in \mathbb{B}^t$ we have that $\langle a, b \odot c\rangle_w  = \langle a \odot b, c \rangle_w$. From the first two properties, we can immediately see that $f_v^+, f_v^-$ are both self-inverse and hence bijective.

\begin{lemma}\label{lem:fv0}
    $f_v^+$ maps proposals of type $T_{i,j}$ to proposals of type $T_{j,i}$ for $i,j \in \{0, 1\}$. Moreover, for any $p \in B^*$ we have that $\langle v, f_v^+(p) \rangle_w = \langle p, p_{IWM}\rangle_w$.
\end{lemma}

\begin{proof}
    Fix some $p \in B^*$. Then:
    \begin{align*}
        \langle v, f_v^+(p) \rangle_w &= \langle v, p \odot v \odot p_{IWM} \rangle_w = \langle v \odot p \odot v, p_{IWM} \rangle\\
        &= \langle p, p_{IWM} \rangle_w
    \end{align*}

    As $f_v^+(p) \in B^*$, using the shown equality with one more application of $f_v^+$ and the fact that $f_v^+$ is self-inverse gives us that $\langle f_v^+(p), p_{IWM} \rangle_w = \langle v, f_v^+(f_v^+(p)) \rangle_w = \langle v, p\rangle_w$. Hence, applying $f_v^+$ has the effect of swapping a proposal's weighted agreement with $v$ with its weighted agreement with $p_{IWM}$. Therefore, $f_v^+$ maps $T_{i,j}$ to $T_{j,i}$.
\end{proof}

\begin{lemma}\label{lem:fv1}
    $f_v^-$ maps proposals of type $T_{i,j}$ to proposals of type $T_{(-j),(-i)}$ for $i,j \in \{-1, +1\}$. For any $p \in B^*$ we have that $\langle v, f_v^-(p)\rangle_w = -\langle p_{IWM}, p\rangle_w$.
\end{lemma}

\begin{proof}
    Fix some $p \in B^*$. Then:
    \begin{align*}
        \langle v, f_v^-(p)\rangle_w &= \langle v, p \odot -(v \odot p_{IWM})\rangle_w = \langle v \odot p \odot -v, p_{IWM}\rangle_w\\
        &= -\langle p, p_{IWM}\rangle_w
    \end{align*}
    Then again, as $f_v^-(p) \in B^*$, using the shown equality with one more application of $f_v^-$ and the fact that its self-inverse gives us that $-\langle f_v^-(p), p_{IWM}\rangle_w = \langle v, f_v^-(f_v^-(p))\rangle_w = \langle v, p \rangle_w$.
Hence we have that $f^-$ swaps and negates a proposal's weights agreement with $v$ with its weighted agreement with $p_{IWM}$. Therefore it indeed maps $T_{i,j}$ to $T_{(-j), (-i)}$.
\end{proof}

\medskip

Now we combine our two bijective maps into a single map $f_v \colon B^* \rightarrow B^*$ defined as follows: \begin{align*}
    f_v(p) =\begin{cases}
        f_v^+(p) &\text{if $p$ is of type $T_{-1,-1}$ or $T_{+1,+1}$}\\
        f_v^-(p) &\text{if $p$ is of type $T_{-1,+1}$ or $T_{+1,-1}$}
    \end{cases}
\end{align*}
It follows from $f_v^+$ and $f_v^-$ being self-inverse and \cref{lem:fv0,lem:fv1} that $f_v$ is also self-inverse.\\

\begin{corollary}\label{cor:fv}
    $f_v$ maps proposals of type $T_{i,j}$ to proposals of type $T_{i,j}$. For any proposal $p$ of type $T_{-1,-1}$ or $T_{+1,+1}$ we have $\langle v, f_v(p)\rangle_w = \langle p, p_{IWM} \rangle_w$ and for any proposal $p$ of type $T_{-1,+1}$ or $T_{+1,-1}$ we have that $\langle v, f_v(p) \rangle_w = -\langle p, p_{IWM} \rangle_w$.
\end{corollary}
\begin{proof}
    This follows directly from the definition of $f_v$ as well as \cref{lem:fv0,lem:fv1}.
\end{proof}

\subsubsection{Thought Experiments}
Now we will detail our two quantities of interest through two thought experiments and show that their expectations are equivalent and non-negative. We define a subset of $\Bset^t$ denoted as $B_m = \{p \in \Bset^t \colon \langle p, p_{IWM}\rangle_w > 0\}$. Equivalently, for all $p \in B_m$, $d_H(p, p_{IWM}, w) < 1/2$. As $w$ is externally fixed, this is some known set. We also note that the number of +1s minus the number of -1s for a given topic $j$ in the preference profile can be written as $n \cdot m_j - n(1-m_j) = n(2m_j - 1)$. Finally, let $R_w$ be the set of all relevant topics under weight vector $w$ (for a discussion of relevant topics, see \cref{app:relevant_topics}).\\

\noindent \textbf{Thought Experiment 1.} Our first thought experiment keeps a global counter $X$ initialized to 0 and samples a proposal $p \in B_m$ uniformly at random. For each voter $i \in [n]$, we add $\langle v_i, p \rangle_w$ to $X$ --- equivalently, for each voter we go topic by topic and add $w_j$ to $X$ for each topic $j$ voter $i$ agrees with $p$ on, and subtract $w_j$ from $X$ for each topic $j$ voter $i$ disagrees with $p$ on. 

\begin{theorem}\label{thm:te1}
    $\Ex[X] \geq 0$. If there exists some $j \in R_{\w}$ such that $m_j > 0.5$ then $\Ex[X] > 0$.
\end{theorem}
\begin{proof} We can write $X$ in terms of variables $X_{i,j}$ that take on 1 if $v_{ij} = p_j$ and -1 otherwise as $X = \sum_{i=1}^n \sum_{j=1}^t w_j X_{i,j}$. Then we can evaluate $\Ex[X]$ more easily:\begin{align*}
    \Ex[X] &= \sum_{i=1}^n \sum_{j=1}^t w_j \Ex[X_{i,j}]
    = \sum_{j=1}^t w_j \sum_{i=1}^n \Ex[X_{i,j}]\\
    &= \sum_{j=1}^t w_j \cdot n(2m_j - 1) (\Pr(p_j=+1) - \Pr(p_j=-1))
\end{align*}
We arrive at the last line because if $p_j = +1$ then $\sum_{i=1}^n X_{i,j} = \sum_{i \in [n] \colon v_{i,j} = +1} 1 + \sum_{i \in [n] \colon v_{i,j} = -1} -1 = n(2m_j - 1)$ and similarly if $p_j = -1$ then $\sum_{i=1}^n X_{i,j} = -n(2m_j - 1)$. Note that we can rewrite the sum with only the terms where there is a strict majority of +1 (so $m_j > 0.5$), because all terms where $m_j = 0.5$ evaluate to 0: \begin{align*}
\Ex[X] &= \sum_{\substack{j\in [t] \\ m_j > 0.5}} w_j n(2m_j - 1) (\Pr(p_j=+1) - \Pr(p_j=-1))\\
&= \sum_{\substack{j\in [t] \\ m_j > 0.5}} w_j n(2m_j - 1) (\Pr(p_j=(p_{IWM})_j) - \Pr(p_j \neq (p_{IWM})_j))
\end{align*}
We arrive at the last line by noting that for all $j$ with $m_j > 0.5$, $(p_{IWM})_j$ must be $+1$ as it is the unique majority for that issue.

Now we observe that $\Pr(p_j = (p_{IWM})_j) \geq \Pr(p_j \neq (p_{IWM})_j)$ for all $j \in [t]$. Fix some $j \in [t]$. Then for any $p \in B_m$ such that $p_j \neq (p_{IWM})_j$, we know that there exists a $p' \in B_m$ that matches $p$ on all entries except $j$, where $p'_j = (p_{IWM})_j$. $p'$ is indeed in $B_m$ because $\langle p', p_{IWM} \rangle_w = \langle p, p_{IWM} \rangle_w + 2w_j > 1/2$. Note that this mapping from $p$ to $p'$ is injective, so we have that there are at least as many proposals in $B_m$ with a $(p_{IWM})_j$ for issue $j$ as there are with a $-(p_{IWM})_j$. Hence, when selecting a proposal from $B_m$ uniformly at random, the probability that it has a $(p_{IWM})_j$ for issue $j$ is at least the probability that it has a $-(p_{IWM})_j$.

We also claim that $\Pr(p_j = (p_{IWM})_j) > \Pr(p_j \neq (p_{IWM})_j)$ if and only if $j \in R_w$. To see a proof of this claim, please refer to \cref{app:relevant_topics} \cref{claim:relevant_probs}. Therefore, we know that for all $j \not \in R_w$, $\Pr(p_j = (p_{IWM})_j) - \Pr(p_j \neq (p_{IWM})_j) = 0$. Then we can rewrite our expectation only in terms of the relevant topics:\begin{align*}
    \Ex[X] &= \sum_{\substack{j\in R_w\\ m_j > 0.5}} w_j n(2m_j - 1) (\Pr(p_j=(p_{IWM})_j) - \Pr(p_j \neq (p_{IWM})_j))
\end{align*}
For any relevant topic $j$, $w_j$ is strictly positive (otherwise $j$ would not be in a minimal topic group) and as just mentioned $\Pr(p_j = (p_{IWM})_j) - \Pr(p_j \neq (p_{IWM})_j)$ is strictly positive as well. $n(2m_j - 1)$ must be positive in all of our terms because we only consider $j \in R_w$ such that $m_j > 0.5$. Therefore, every term in the sum is positive. If no relevant topic has $m_j > 0.5$ then $\Ex[X] = 0$. Otherwise, we have that there exists at least one term left in the sum, and $\Ex[X] > 0$.
\end{proof}

\noindent \textbf{Thought Experiment 2.} For our second thought experiment we again sample $p \in B_m$ uniformly at random and maintain a global counter $Y$ initialized to 0. Each voter $i$ compares $v_i$ with $p$. If they approve of $p$ then they add $\langle p, p_{IWM} \rangle_w$ to $Y$, and if they disapprove of $p$ then they subtract $\langle p, p_{IWM} \rangle_w$ from $Y$. If they are neutral (so $\hamdist(v_i, p, w) = 1/2$) then they leave $Y$ unchanged.\\

\begin{theorem}\label{thm:te1_is_te2}
    $\Ex[X] = \Ex[Y]$
\end{theorem}
\begin{proof}
    First we write $Y = \sum_{i \in [n]} Y_i$ where $Y_i$ is $\langle p, p_{IWM} \rangle_w$ if voter $i$ approves of $p$, $-\langle p, p_{IWM} \rangle_w$ if voter $i$ disapproves of $p$, and 0 if voter $i$ is neutral. Then $\Ex[Y] = \sum_{i \in [n]} \Ex[Y_i]$ by linearity of expectation. From the first thought experiment, we have that $\Ex[X] = \sum_{i \in [n]} \Ex[\langle v_i, p \rangle_w]$ by definition of $X$. Hence, to show that $\Ex[X] = \Ex[Y]$, it suffices to show that for all $i \in [n]$ we have that $\Ex[Y_i] = \Ex[\langle v_i, p \rangle_w]$.

    Fix some voter $i \in [n]$. Let $B_+$ and $B_-$ be the sets of proposals in $B_m$ that voter $i$ approves of and disapproves of, respectively. Note that if $i$ is neutral about $B_m$ then $\langle v_i, p\rangle_w = 0$. Then we have that:\begin{align*}
        \Ex[\langle v_i, p\rangle_w] &= |B_m|^{-1} \sum_{p \in B_m} \langle v_i, p\rangle_w\\
        &= |B_m|^{-1} \left(\sum_{p \in B_+} \langle v_i, p\rangle_w + \sum_{p \in B_-} \langle v_i, p\rangle_w\right)
    \end{align*}
    Then, because $f_{v_i}$ is self-inverse, we can write:
    \begin{align*}
        \Ex[\langle v_i, p\rangle_w] &= |B_m|^{-1} \left(\sum_{p \in B_+} \langle v_i, f_{v_i}(f_{v_i}(p))\rangle_w + \sum_{p \in B_-} \langle v_i, f_{v_i}(f_{v_i}(p))\rangle_w\right)\\
        &= |B_m|^{-1} \left(\sum_{p \in B_+} \langle f_{v_i}(p), p_{IWM} \rangle_w - \sum_{p \in B_-} \langle f_{v_i}(p), p_{IWM} \rangle_w\right)
    \end{align*}
    In the last line we use the fact that for $p \in B_+$, $p \in T_{+1,+1}$ because $p \in B_m$ so $\langle p, p_{IWM} \rangle_w > 0$ and $p \in B_+$ implies that $\langle v_i, p \rangle_w > 0$. By \cref{cor:fv}, we know that $f_{v_i}(p)$ is then also in $T_{+1,+1}$, and hence that $\langle v_i, f_{v_i}(f_{v_i}(p))\rangle_w = \langle f_{v_i}(p), p_{IWM} \rangle_w$. Similarly, we have that for $p \in B_-$, $p \in T_{+1,-1}$ because $p \in B_m$ so $\langle p, p_{IWM}\rangle_w > 0$ and $p \in B_-$ implies that $\langle v_i, p\rangle_w < 0$. Then $f_{v_i}(p)$ is also in $T_{+1,-1}$, so $\langle v_i, f_{v_i}(f_{v_i}(p))\rangle_w = -\langle f_{v_i}(p), p_{IWM} \rangle_w$. As $f_{v_i}$ is a bijection on both $B_+$ and $B_-$, summing over a function of $f_{v_i}(p)$ for all $p \in B_+ (B_-)$ is equivalent to summing over that function of $p$ for all $p \in B_+ (B_-)$:\begin{align*}
        \Ex[\langle v_i, p\rangle_w] &= |B_m|^{-1} \left(\sum_{p \in B_+} \langle p, p_{IWM} \rangle_w - \sum_{p \in B_-} \langle p, p_{IWM} \rangle_w\right)\\
        &= \Ex[Y_i]
    \end{align*}
    Therefore, we have shown that $\Ex[X] = \Ex[Y]$.
\end{proof}

\noindent Finally, we can prove~\cref{thm:external_iwm}:
\begin{proof}
    By \cref{thm:te1_is_te2}, we know that $\Ex[Y] = \Ex[X]$, and by \cref{thm:te1} we have that this quantity is always non-negative and is strictly positive when in there exists some relevant topic $j$ such that $m_j > 0.5$. We call this condition the strict majority guarantee. Then there exists some proposal $p \in B_m$ such that $Y \geq 0$, or $Y >0$ if we have the strict majority guarantee.

    We now rewrite $Y$ in a way that makes the connection to majority-supported proposals more explicit. We can write $Y = \langle p, p_{IWM} \rangle_w \cdot ((\text{\# voters supporting } p) - (\text{\# voters opposing } p))$. As $p \in B_m$, we know that $\langle p, p_{IWM} \rangle_w > 0$. Hence we have that:\begin{align*}
       (\text{\# voters supporting } p) - (\text{\# voters opposing } p) &= \frac{Y}{\langle p, p_{IWM}\rangle}\\
        \geq 0
    \end{align*}
    where the inequality is strict if we have the strict majority guarantee. As the number of approving voters is at least the number of disapproving voters, $p$  is weakly majority-supported. With the strict majority guarantee, $Y> 0$ so we have that the number of approving voters is strictly greater than the number of disapproving voters, and $p$ is strictly majority-supported.
\end{proof}

\subsection{Internal Weights: Non-Losing Proposals}\label{app:internal_iwm}

In this section of the appendix, we provide proofs and further discussion of the results presented in \cref{subsec:anscombe_internal}. In \cref{app:positive_internal_iwm} we provide intuition and proof for our upper bounds on $g_{\ell}$ in \cref{thm:indiv_weights_positive}, and in \cref{app:negative_internal_iwm} we provide proof for our lower bounds on $g_{\ell}$ in \cref{thm:indiv_weights_negative}. Finally in \cref{app:wagner} we prove the generalized Rule of Three-Fourths. Throughout this section of the appendix, we use the notation that $w(S) = \sum_{j \in S} w_j$ for any $S \subseteq [t]$.

\subsubsection{Efficiently Finding Reasonable Non-Losing Proposals in Individual Weight Setting}\label{app:positive_internal_iwm}

As discussed in the body, for any proposal $p \in \mathbb{B}^t$, at least one of $p, \overline{p}$ is weakly majority-supported. Therefore, one way to get a simple upper bound on $g_{\ell}$ is to construct $p$ given any $p_{IWM}$ and take $\max\{d_H(p, p_{IWM}, \tilde{w}), d_H(\overline{p}, p_{IWM}, \tilde{w})\}$. Hence, to get a tighter upper bound, we would like to construct $p$ such that this quantity is small. As $d_H(\overline{p}, p_{IWM}, \tilde{w}) = 1 - d_H(p, p_{IWM}, \tilde{w})$, this problem corresponds to the partition problem, because we want to partition the topics into two sets such that their average weight sums are as close as possible. We formalize this intuition below.

\internaliwmpos*

\begin{proof}
    Fix some instance $\votinginstance = (\prefprof, W)$ with average weight vector $\tilde{w}_{\max}$. We will case on $\tilde{w}_{max}$ and construct a subset of topics, $T$, for later use in creating the proposal complement pair of interest.
    \begin{itemize}
    \item If $\tilde{w}_{\max} \in (0, 1/3)$, construct a set $S \subseteq [t]$ by starting with $S = [t]$ and then removing topics until we cannot remove anymore without the total weight going below $1/2$. We know that $\tilde{w}(S) \leq 1/2 + \tilde{w}_{max}$, as otherwise we could remove another topic without the weight dipping below $1/2$ (as all topics have weight at most $\tilde{w}_{max}$). If  $\tilde{w}(S) \leq 1/2 + \tilde{w}_{max}/2$ then let $T = S$. Otherwise, we have that $\tilde{w}(S) > 1/2 + \tilde{w}_{max}/2$. Take any topic $j \in S$ and let $S' := S \setminus \{j\}$ and $S'' = [t] \setminus S'$. Then $\tilde{w}(S') = \tilde{w}(S) - \tilde{w}_j$ and $\tilde{w}(S'') = t- \tilde{w}(S')$. By construction of $S$ we know that $\tilde{w}(S') < 1/2$. We can also lower bound it by using our lower bound on $\tilde{w}(S)$ and upper bounding $\tilde{w}_j$ by $\tilde{w}_{max}$: $\tilde{w}(S') \geq 1/2 + \tilde{w}_{max}/2 - \tilde{w}_{max} = 1/2 - \tilde{w}_{max}/2$. Then we have that $\tilde{w}(S'') \geq 1/2$ and $\tilde{w}(S'') \leq 1/2 + \tilde{w}_{max}/2$. Then let $T = S''$. Hence, in all cases $1/2 \leq \tilde{w}(T) \leq 1/2 + \tilde{w}_{max}/2$.

    \item If $\tilde{w}_{max} \in [1/3, 1/2]$, let $T = [t] \setminus \{j_{max}\}$ where $j_{max} \in [t]$ is the index of some maximum weight topic. Then $1/2 \leq \tilde{w}(T) = 1-\tilde{w}_{max}$.

    \item If $\tilde{w}_{max} \in (1/2, 1)$ then let $T = \{j_{max}\}$. Then $1/2 \leq \tilde{w}(T) =\tilde{w}_{max}$.
    \end{itemize}
     We construct a proposal $p$ as follows: for all $j \in T$, set $p_j = (p_{IWM})_j$, and for all $j \in [t] \setminus T$ set $p_j = -(p_{IWM})_j$. Then we have that $p$ agrees with $p_{IWM}$ on all topics in $T$ and disagrees with $p_{IWM}$ on all other topics. Then $d_H(p, p_{IWM}, \tilde{w}) = 1-\tilde{w}(T)$ and $d_H(\overline{p}, p_{IWM}, \tilde{w}) = \tilde{w}(T)$. We know that at least one of these is weakly majority-supported. Hence, \begin{align*}
        \min_{p \text{ weakly majority-supported}} d_H(p, p_{IWM}, \tilde{w}) &\leq \max\{\tilde{w}(T), 1-\tilde{w}(T)\}\\
        &= \tilde{w}(T)
    \end{align*}
    We arrive at the second line because $\tilde{w}(T) \geq 1/2$. This bound holds for any selection of $\votinginstance$ with the same maximum average weight and for any selection of $p_{IWM}$ for $\votinginstance$. Hence, $g_{\tilde{w}_{max}} \leq \tilde{w}(T)$, which gives us the desired bounds. Additionally, either $p$ or $\overline{p}$ are weakly majority-supported and have the desired distance from $p_{IWM}$. Determining $T$, constructing $p, \overline{p}$, and checking their support takes time $O(nt)$ altogether.  
\end{proof}

We note that if one wants to find a weakly majority-supported proposal with distance as close as possible to $1/2$ from a designated $p_{IWM}$ using the proposal complement pair technique, this can also be done in polynomial time with an additional assumption on the average weight vector. Specifically, we assume that $W$ has \emph{polynomial precision}, so all of its elements can be expressed as rational numbers with a common denominator that is polynomial in $n$ and $t$. We first scale up all elements of $\tilde{w}$ by the common denominator, so that they are all integers. By our assumption, we know that their scaled-up counterparts are polynomial in $n$ and $t$. We then use the standard reduction from partition to knapsack and run the pseudo-polynomial DP for knapsack. Finally, we check which one of $p$ or $\overline{p}$ (or both) has weak majority support.

\subsubsection{Arbitrarily Low Agreement with the IWM}\label{app:negative_internal_iwm}

We provide the full proof for \cref{thm:indiv_weights_negative} here:
\internalnegative*
\begin{proof}
We will actually show something slightly stronger than just upper bounding $g_{\ell}$ in this proof. In both ranges of $\ell$, our constructions will be such that $\mathbf{1}$ is the $\emph{unique}$ IWM for the instance. Hence, the fact that any non-losing proposal is far away from $\mathbf{1}$ not only implies the bound on $g_{\ell}$, it also implies that there are instances where any non-losing proposal is far away from \emph{any} IWM for that instance.

 \textbf{Small $\boldsymbol{\ell}$.} Fix some $k \in \mathbb{N}^+$ and let $\ell = 1/(2k+1)$. We illustrate the construction below for $k = 1$ for clarity and then describe the generalization to larger $k$.

 \[
\prefprof =\begin{matrix}5\  \times \\ 5 \ \times  \\ 5 \ \times \\ 4 \ \times \end{matrix}\begin{bmatrix}
\tikzmarknode{p411}{\texttt{-1}} & \tikzmarknode{p412}{\texttt{+1}} & \tikzmarknode{p413}{\texttt{+1}} \\
\tikzmarknode{p421}{\texttt{+1}} & \tikzmarknode{p422}{\texttt{-1}} & \tikzmarknode{p423}{\texttt{+1}} \\
\tikzmarknode{p431}{\texttt{+1}} & \tikzmarknode{p432}{\texttt{+1}} & \tikzmarknode{p433}{\texttt{-1}} \\
\tikzmarknode{p441}{\texttt{+1}} & \tikzmarknode{p442}{\texttt{+1}} & \tikzmarknode{p443}{\texttt{+1}} \\
\end{bmatrix} \quad
\mathcal{W} =\begin{matrix}5\  \times \\ 5 \ \times  \\ 5 \ \times \\ 4 \ \times \end{matrix}\begin{bmatrix}
 \tikzmarknode{p411w}{3/5} & \tikzmarknode{p412w}{1/5} & \tikzmarknode{p413w}{1/5}\\
 \tikzmarknode{p421w}{1/5} & \tikzmarknode{p422w}{3/5} & \tikzmarknode{p423w}{1/5}\\
 \tikzmarknode{p431w}{1/5} & \tikzmarknode{p432w}{1/5} & \tikzmarknode{p433w}{3/5}\\
 \tikzmarknode{p441w}{1/3} & \tikzmarknode{p442w}{1/3} & \tikzmarknode{p414w}{1/3}
\end{bmatrix}
\]%
\begin{tikzpicture}[overlay,remember picture, shorten >=-3pt, shorten <= -3pt]
\drawcell{p411}{\czero}
\drawcell{p412}{\cone}
\drawcell{p413}{\cone}
\drawcell{p421}{\cone}
\drawcell{p422}{\czero}
\drawcell{p423}{\cone}
\drawcell{p431}{\cone}
\drawcell{p432}{\cone}
\drawcell{p433}{\czero}
\drawcell{p441}{\cone}
\drawcell{p442}{\cone}
\drawcell{p443}{\cone}
\end{tikzpicture}%

The generalization is as follows: we have $t = 2k+1$ topics and $t+1$ types of voters. Denote type $i$'s preference and weight vectors as $v^i$ and $w^i$ respectively. There are $2t-1$ copies of each of the first $t$ types, and $t+1$ copies of the last type of voter. The voters of type $t+1$ prefer the all-ones vector and consider all issues to be equally important: $v^{t+1} = \boldsymbol{+1}$ and $w^{t+1} = (1/t, 1/t, \ldots, 1/t)$. Voters of type $i \in \{1, \ldots, t\}$ are single-issue voters on issue $i$ and prefer the negative outcome, although they do place some importance on the other issues:\[
v^i_j = \begin{cases}
    -1 &\text{if } i = j\\
    +1 &\text{o.w.}
\end{cases} \quad w^i_j = \begin{cases}
\frac{t}{2t-1} &\text{if } i = j\\
\frac{1}{2t-1} &\text{o.w.}
\end{cases}
\]
We say that they are single-issue voters because they vote for a proposal if and only if it agrees with their position on that issue. Note that by symmetry of the weight matrix, all topics have the same weight in the average weight vector. As there are $t = 2k+1$ topics, $\tilde{w}_j = 1/(2k+1) = \ell$ for all $j \in [t]$. Therefore, $\ell$ is indeed $\tilde{w}_{max}$ of this profile. Additionally, all of the weights are nonnegative, and every weight vector type sums to 1, as $\frac{t}{2t-1} + (t-1) \cdot \frac{1}{2t-1} = 1$.

Now we show that $\mathbf{+1}$ is the unique IWM for this preference and weight profile. For any given topic, the total weight on $+1$ is $(t-1) \cdot (2t-1) \cdot \frac{1}{2t-1} + (t+1) \cdot \frac{1}{t}$, and the weight on $-1$ is $(2t-1) \cdot \frac{t}{2t-1} = t$. We have that $(t-1) \cdot (2t-1) \cdot \frac{1}{2t-1} + (t+1) \cdot \frac{1}{t} = t + \frac{1}{t} > t$. Therefore, $+1$ is the strict majority opinion on all topics.

Fix any $p \in \mathbb{B}^t$ such that $d_H(p, \mathbf{1}, \tilde{w}) < 1/2 + \ell/2$. Given our construction, this means that at least $k+1$ of the $2k+1$ topics are set to $+1$ in the proposal. Note that for all $i \in [t]$ such that $p_i = +1$, all voters of type $i$ vote no on $p$. As there are at least $k+1$ of these indices, we have that at least $k+1$ types of voters vote against $p$. This corresponds to at least $(2t-1) \cdot (k+1)$ voters. The proposal can then get at most all of the remaining votes, which amount to $(2t-1) \cdot k + (t+1)$ votes. We have that $(2t-1)(k+1) = (2t-1)k + (2t-1) > (2t-1)k + (t+1)$ because $t \geq 3$. Therefore, $p$ receives strictly more votes against it than for it. Hence, any non-losing proposal must have distance at least $1/2 + \ell/2$ from $\mathbf{1}$. 
 
 \textbf{Big $\boldsymbol{\ell}$.} Fix some $\ell \in (1/2, 1)$. We will construct a preference and weight profile, $\prefprof, \mathcal{W}$, such that $\tilde{w}_{max} = \ell$, $\mathbf{1}$ is the sole IWM, and any non-losing proposal $p$ has $d_H(p, \mathbf{1}, \tilde{w}) \geq \ell$. Let $x \in \mathbb{N}^+$ be such that $x > \max \left\{ \frac{\ell}{1-\ell}, \frac{1}{2\ell - 1}\right\}$. Note that both of the denominators are strictly positive because of our bounds on $\ell$.
\[
\prefprof =\begin{matrix}x\  \times \\ x+1 \ \times  \end{matrix}\begin{bmatrix}
\tikzmarknode{p311}{\texttt{+1}} & \tikzmarknode{p312}{\texttt{+1}} \\
\tikzmarknode{p321}{\texttt{-1}} & \tikzmarknode{p322}{\texttt{+1}} 
\end{bmatrix} \quad
\mathcal{W} =\begin{matrix}x \ \times \\ x+1 \  \times \end{matrix}\begin{bmatrix}
 \tikzmarknode{p311w}{\frac{x+1}{x} \cdot \ell} & \tikzmarknode{p312w}{1 - \frac{x+1}{x} \cdot \ell}\\
 \tikzmarknode{p321w}{\frac{x}{x+1} \cdot \ell} & \tikzmarknode{p322w}{1 - \frac{x}{x+1} \cdot \ell} 
\end{bmatrix}
\]%
\begin{tikzpicture}[overlay,remember picture, shorten >=-3pt, shorten <= -3pt]
\drawcell{p311}{\cone}
\drawcell{p312}{\cone}
\drawcell{p321}{\czero}
\drawcell{p322}{\cone}
\end{tikzpicture}%

 First we show that all elements of $\mathcal{W}$ are in $[0,1]$. It suffices just to show that this is true for the weights on the first issue, as every row sums to 1 and there are only two issues. More specifically, we will show all voters place weight in the range $(0.5, 1)$ on the first topic. First we upper bound:\begin{align*}
 \frac{\ell(x+1)}{x} < 1 \iff \ell x + \ell < x \iff \ell < x(1-\ell) \iff\frac{\ell}{1-\ell} < x
 \end{align*}
where the final inequality holds by our definition of $x$. Additionally,\[
\frac{\ell x}{x+1} > 0.5 \iff 2\ell x > x+1 \iff x > \frac{1}{2\ell - 1}
\]
where again the final inequality holds by our definition of $x$. Putting everything together we have that: \[
0.5 < \frac{\ell x}{x+1} \leq \frac{\ell (x+1)}{x} < 1
\]
Not only does this confirm that $\mathcal{W}$ is a valid weight profile, it also informs us that both types of voters in this scenario are single issue voters on the first issue --- they vote for a proposal if and only if it agrees with their preference on the first issue. Additionally, it gives us that $\tilde{w}_1 = \tilde{w}_{max}$. We have that \begin{align*}
\tilde{w}_1 &= \left(x \cdot \frac{\ell(x+1)}{x} + (x+1) \cdot \frac{\ell x}{x+1}\right) \frac{1}{2x+1}\\
&= \frac{\ell (x+1) + \ell x}{2x+1} = \ell
\end{align*}
Hence, $\tilde{w}_{max} = \ell$ as desired.

Next, we show that $\mathbf{1}$ is the sole IWM for this profile. Clearly $+1$ is the unanimous majority opinion for the second topic. For the first topic, we have that $+1$ is the unique majority opinion if there is strictly more total weight on $+1$ than on $-1$ for that topic:\[
x \cdot \frac{\ell(x+1)}{x} > (x+1) \cdot \frac{\ell x}{x+1} \iff x+1 > x
\]
Therefore, $+1$ is the strict majority for both issues and as such $\mathbf{1}$ is the unique IWM for this profile. 

Finally, we claim that any non-losing proposal $p$ must have $p_1 = -1$. To see this, recall that all voters are single-issue voters on the first topic. All $x+1$ voters with $-1$ as their preference for the first topic would vote against $p$ if $p_1 = +1$. As they form a majority of voters, the proposal would lose. Therefore, $d_H(p, \mathbf{1}, \tilde{w}) \geq \tilde{w}_1 = \ell$.

It is also worth noting 
that this preference profile is single-switch, but Ostrogorski's paradox happens: $(-1, +1) \succ_\votinginstance (+1, +1)$. This instance highlights that the single-switch condition \emph{does not} help for the internal weights setting.
\end{proof}

\subsection{Internal Weights: Condition Precluding Anscombe's Paradox}\label{app:wagner}

Here we include the full proof for our generalized Rule of Three-Fourths, \cref{thm:three-fourths}:
\threefourths*
\begin{proof}
    Fix an instance $\votinginstance = (\prefprof, W)$ in the internal weights setting such that $\tilde{m} \geq 3/4$. Notice if $W$ has identical rows, then this is equivalent to the external weights setting. We assume without loss of generality that the IWM proposal we are interested in verifying gets weak majority support is $\mathbf{1}$. This is indeed without loss of generality because if our original $p_{IWM}$ has a $-1$ for some topic, $j$, we know that $m_j = 0.5$ (as we assume for all of \cref{sec:anscombe} that $m_j \geq 0.5$, so for $-1$ to be a majority, the column must be exactly split). We can then flip all entries in that column of the preference profile --- this is equivalent to having voters express their preferences on the negated version of the issue. Then $(p_{IWM})_j = +1$ as well because its decision on the \emph{negated} version of issue $j$ is the opposite of its former decision. Moreover, the fraction of weight on $+1$ in that column is still $0.5$. So all issue majorities are unchanged, and hence the average majority is also unchanged. Therefore, assume $p_{IWM} = \mathbf{1}$.
    
    We first define a variable $\wones$ counting the total weight placed on $+1$ in a preference profile, and show that if $\wones \geq \frac{3n}{4}$ then $\mathbf{1}$ is weakly majority-supported (hence Anscombe's paradox does not occur). Then we will show that $\wones = n \cdot \avgmaj$.\\

    Let $\wones := \sum_{i=1}^n 1 - d_H(v_i, \mathbf{1}, w_i)$. We claim that if \begin{align*}
        \wones \geq \left(\left\lfloor \frac{n}{2}\right\rfloor + 1\right)\left(\frac{1}{2}\right) + \left\lfloor\frac{n-1}{2}\right\rfloor
    \end{align*}
    then $\mathbf{1}$ is weakly majority-supported. Assume for the sake of contradiction that this is not the case. For any voter $i$ that opposes $\mathbf{1}$, we have that $d_H(v_i, \mathbf{1}, w_i) > 1/2$. As $\mathbf{1}$ is not even weakly majority-supported, we know that more than half of the voters (at least $\lfloor n/2 \rfloor + 1$) oppose $\mathbf{1}$. The remaining at most $n - \lfloor n/2 \rfloor -1 = \lfloor (n-1)/2 \rfloor$ voters still must have a non-negative distance from $\mathbf{1}$. Hence we can upper bound $\wones$:\begin{align*}
        W_{ones} &= \sum_{\substack{i \in [n] \\ i \text{ opposes } \mathbf{1}}} 1-d_H(v_i, \mathbf{1}, w_i) + \sum_{\substack{i \in [n] \\ i \text{ supports } \mathbf{1}}} 1-d_H(v_i, \mathbf{1}, w_i)\\
        &< \sum_{\substack{i \in [n] \\ i \text{ opposes } \mathbf{1}}} 1/2 + \sum_{\substack{i \in [n] \\ i \text{ supports } \mathbf{1}}} 1\\
        &\leq \left(\left\lfloor \frac{n}{2}\right\rfloor + 1\right)\left(\frac{1}{2}\right) + \left\lfloor\frac{n-1}{2}\right\rfloor
    \end{align*}
    
    Hence, we have that $\wones < \left(\left\lfloor \frac{n}{2}\right\rfloor + 1\right)\left(\frac{1}{2}\right) + \left\lfloor\frac{n-1}{2}\right\rfloor$, a contradiction. We now upper bound the RHS as follows:\begin{align*}
        \left(\left\lfloor \frac{n}{2}\right\rfloor + 1\right)\left(\frac{1}{2}\right) + \left\lfloor\frac{n-1}{2}\right\rfloor &\leq \left(\frac{n}{2}+1\right)\frac{1}{2} + \left(\frac{n-1}{2}\right) = \frac{3n}{4}
    \end{align*}
    Hence, if $\wones \geq \frac{3n}{4}$ then our previous condition is satisfied, and the issue-wise majority is non-losing.\\

    Now we'll show the claimed relationship between $\avgmaj$ and $\wones$:\begin{align*}
        n\cdot \avgmaj &= n \sum_{j=1}^t \avgweight_j m_j = n\sum_{j=1}^t \avgweight_j \left(\frac{1}{n\avgweight_j}\sum_{i=1}^n w_{i,j} \cdot \mathbb{I}(v_{i,j} = +1)\right)\\
        &= \sum_{j=1}^t \sum_{i=1}^n w_{i,j} \cdot \mathbb{I}(v_{i,j} = +1) = \wones
    \end{align*}
    Therefore we have that $\wones \geq \frac{3n}{4} \iff \avgmaj \geq \frac{3}{4}$.

    To prove the second claim of the theorem, we fix a new instance $\mathcal{I'} = (\mathcal{P}, w)$ in the external weights model such that $m_j \geq 3/4$ for all $j \in [t]$. Therefore, $\boldsymbol{1}$ is the unique IWM. Assume for sake of contradiction that Ostrogorski's paradox occurs, so there is some proposal $p \neq \mathbf{1}$ such that $p \succ_{\votinginstance'} \mathbf{1}$. Then we know by \cref{lem:ansc_ost_relation} that there exists a sub-instance $\votinginstance''$ in which Anscombe's paradox occurs, where $\votinginstance''$ is obtained by restricting $\mathcal{I'}$ to some subset of issues $T \subseteq [t]$ and renormalizing the external weight vector. Note that the majorities on the topics in $T$ are unchanged from the original profile. Hence, the average majority in $\votinginstance'$ is at least $3/4$, because each individual issue majority is at least $3/4$. This is a contradiction to the claim proven above. Therefore, Ostragorski's paradox does not occur in $\mathcal{I'}$.
\end{proof}

\end{document}